\documentclass[a4paper]{article}
\usepackage{graphicx}
\usepackage{subfigure}
\usepackage{url}
\usepackage{textcomp}
\usepackage{fancyhdr}
\usepackage{appendix}
\usepackage{amsmath}
\usepackage{amssymb}
\usepackage{amsthm}
\usepackage{bbm}
\usepackage{siunitx}

\usepackage{tikz}
\usetikzlibrary{arrows.meta,positioning,calc,fit}
\definecolor{accent}{RGB}{33,76,118}
\definecolor{soft}{RGB}{238,244,249}

\graphicspath{
{/Users/bertkappen/doc/publ/},
{/Users/bertkappen/doc/publ/thoughts/}
{/Users/bertkappen/doc/publ/thoughts/non_backprop_learning/}
{/Users/bertkappen/doc/publ/nips2013_workshop/}
{/Users/bertkappen/doc/publ/control_quantum/control_open_quantum_systems/position_measurement/}
{/Users/bertkappen/doc/publ/control_quantum/control_open_quantum_systems/simulations/numerical_accuracy/}
{/Users/bertkappen/doc/publ/control_quantum/control_open_quantum_systems/simulations/}
{/Users/bertkappen/doc/figures/screenshots/}
}

\usepackage[english]{babel}
\usepackage{lscape}
\usepackage{appendix}
\usepackage{amsmath}
\usepackage{amssymb}
\usepackage{amsthm}
\usepackage{txfonts}
\usepackage{graphicx}
\usepackage[all]{xy}
\usepackage{latexsym}
\usepackage{url}
\usepackage{subfigure}
\usepackage{movie15}
\usepackage{hyperref}
\usepackage{algpseudocode}
\usepackage{multicol}
\usepackage{hyperref}
\usepackage{wrapfig}

\newtheorem{thrm}{Theorem}

\newtheorem{lemm}[thrm]{Lemma}
\newtheorem{coro}[thrm]{Corollary}

\newcommand{\bmpt}[1]{\begin{minipage}[t]{#1\textwidth}}
\newcommand{\bmp}[1]{\begin{minipage}{#1\textwidth}}
\newcommand{\emp}{\end{minipage}}

\newcommand{\cA}{{\cal A}}

\newcommand{\cD}{{\cal D}}
\newcommand{\cE}{{\cal E}}

\newcommand{\cG}{{\cal G}}

\newcommand{\cJ}{{\cal J}}
\newcommand{\cK}{{\cal K}}
\newcommand{\cL}{{\cal L}}
\newcommand{\cM}{{\cal M}}
\newcommand{\cN}{{\cal N}}

\newcommand{\cP}{{\cal P}}

\newcommand{\cR}{{\cal R}}
\newcommand{\cS}{{\cal S}}

\newcommand{\bv}{\begin{verbatim}}
\newcommand{\ev}{\end{verbatim}}

\newcommand{\E}{\ensuremath{\mathbb{E}}}

\newcommand{\C}{\ensuremath{\mathbb{C}}}

\newcommand{\R}{\ensuremath{\mathbb{R}}}

\newcommand{\Tr}{\mathrm{Tr}}
\newcommand{\ben}{\begin{enumerate}}
\newcommand{\een}{\end{enumerate}}
\newcommand{\be}{\begin{equation}}
\newcommand{\ee}{\end{equation}}
\newcommand{\bea}{\begin{eqnarray}}
\newcommand{\eea}{\end{eqnarray}}
\newcommand{\beaa}{\begin{eqnarray*}}
\newcommand{\eeaa}{\end{eqnarray*}}
\newcommand{\ba}{\begin{align}}
\newcommand{\ea}{\end{align}}
\newcommand{\baa}{\begin{align*}}
\newcommand{\eaa}{\end{align*}}
\newcommand {\bd}{\begin{description}}
\newcommand {\ed}{\end{description}}
\newcommand{\bi}{\begin{itemize}}
\newcommand{\ei}{\end{itemize}}
\newcommand{\bc}{\begin{center}}
\newcommand{\ec}{\end{center}}
\newcommand {\bq}{\begin{quotation}}
\newcommand {\eq}{\end{quotation}} 
\newcommand {\bdo}{\begin{document}} 
\newcommand {\edo}{\end{document}} 
\newcommand{\eql}{\: = \:}
\newcommand{\dfl}{\: \equiv \:}
\newcommand{\onder}[2]{#1_{\mbox{\scriptsize #2}}}
\newcommand {\boven}[2]{#1^{\mbox{\scriptsize #2}}}
\newcommand{\argmin}[1]{\mathop{\rm argmin}\limits_{#1}}
\newcommand{\klein}[2]{#1_{\mbox{\tiny #2}}}
\newcommand{\doo}{\mbox{do}}
\newcommand{\fracd}[2]{{{\displaystyle #1} \over {\displaystyle #2}}}

\newcommand{\sm}[2]{#1 _{\setminus #2}}
\newcommand{\smi}[1]{\sm{#1}{i}}
\newcommand{\smk}[1]{\sm{#1}{k}}
\newcommand{\intarg}[2]{\!#2\,d#1}
\newcommand{\del}[3]{\left#1 #3 \right#2}
\newcommand{\av}[1]{\del{<}{>}{#1}}
\newcommand{\avv}[1]{\del{<<}{>>}{#1}}
\newcommand{\bra}[1]{\del{<}{|}{#1}}
\newcommand{\ket}[1]{\del{|}{>}{#1}}
\newcommand{\avmini}[1]{\del{<}{>_{\setminus i}}{#1}}
\newcommand{\avminij}[1]{\del{<}{>_{\setminus i,j}}{#1}}
\newcommand{\sminij}{s_{\setminus i,j}}
\newcommand{\ui}{\avmini{u_i^2}}
\newcommand{\hi}{\avmini{h_i}}
\newcommand{\br}[1]{\del{\{}{\}}{#1}}
\newcommand{\bl}[1]{\del{[}{]}{#1}}

\newcommand {\Hxx}{H^{(xx)}}
\newcommand {\Hxy}{H^{(xy)}}
\newcommand {\Hyx}{H^{(yx)}}
\newcommand {\Hyy}{H^{(yy)}}
\newcommand {\Cxx}{C^{(xx)}}
\newcommand {\Cxy}{C^{(xy)}}
\newcommand {\Cyx}{C^{(yx)}}
\newcommand {\Cyy}{C^{(yy)}}
\newcommand {\cx}{c^{(x)}}
\newcommand {\hy}{h^{(y)}}
\newcommand{\bt}{\bar{t}}
\newcommand {\bx}{\bar{x}}
\newcommand {\bX}{\bar{X}}
\newcommand {\by}{\bar{y}}
\newcommand {\bw}{\bar{w}}
\newcommand {\bh}{\bar{h}}
\newcommand {\vs}{\vec{s}}
\newcommand {\vbfs}{{\bf s}}
\newcommand {\vv}{\vec{v}}
\newcommand {\vlambda}{\vec{\lambda}}
\newcommand {\vh}{\vec{h}}
\newcommand {\vk}{\vec{k}}
\newcommand {\vz}{\vec{z}}
\newcommand {\vx}{\vec{x}}
\newcommand {\valpha}{\vec{\alpha}}
\newcommand {\vm}{\vec{m}}
\newcommand {\vmu}{\vec{\mu}}
\newcommand {\vtheta}{\vec{\theta}}
\newcommand {\vgamma}{\vec{\gamma}}
\newcommand {\qd}{q_{\partial_i}}
\newcommand {\qm}{q_{\setminus i}}
\newcommand{\vxv}{\vec{x}_v}
\newcommand{\vxh}{\vec{x}_h}
\newcommand {\bfK}{{\bm{K}}}
\newcommand {\bfb}{{\bf b}}
\newcommand {\vy}{\vec{y}}
\newcommand {\bfy}{{\bf y}}
\newcommand {\bfx}{{\bf x}}
\newcommand {\bfA}{{\bm{A}}}
\newcommand {\bfX}{{\bm{X}}}
\newcommand {\bfC}{{\bm{C}}}
\newcommand {\bfH}{{\bf H}}
\newcommand {\bfL}{{\bm{L}}}
\newcommand {\byfI}{{\bm{I}}}
\newcommand {\bfS}{{\bm{S}}}
\newcommand {\bft}{{\bf t}}
\newcommand {\bfm}{{\bf m}}
\newcommand {\bfu}{{\bf u}}
\newcommand {\bfw}{{\bf w}}
\newcommand {\bfW}{{\bf W}}
\newcommand {\bfmu}{{\bf \mu}}
\newcommand {\bfphi}{{\bm{\phi}}}
\newcommand {\bftheta}{{\bm{\theta}}}
\newcommand {\bfeta}{{\bm{\eta}}}
\newcommand {\bfchi}{{\bm{\chi}}}
\newcommand {\bfPhi}{{\bm{\Phi}}}
\newcommand {\bfLambda}{{\bm{\Lambda}}}
\newcommand {\bfSigma}{{\bm{\Sigma}}}
\newcommand {\vw}{\vec{w}}
\newcommand {\vb}{\vec{b}}
\newcommand {\vu}{\vec{u}}
\newcommand {\vW}{\vec{W}}
\newcommand {\vM}{\vec{M}}
\newcommand {\sign}{\mbox{sign}}
\newcommand {\sech}{\mbox{sech}}
\newcommand {\hmu}{h^{\mu}_j}
\newcommand {\gmu}{g^{\mu}_j}
\newcommand {\gjkmu}{g^{\mu}_{jk}}
\newcommand {\sia}{s_i^\alpha}
\newcommand {\sjb}{s_j^\beta}
\newcommand {\hia}{h_i^\alpha}
\newcommand {\hjb}{h_j^\beta}
\newcommand {\mia}{m_i^\alpha}
\newcommand {\mib}{m_i^\beta}
\newcommand {\mjb}{m_j^\beta}
\newcommand {\wijab}{w_{ij}^{\alpha\beta}}
\newcommand {\chiijab}{\chi_{ij}^{\alpha\beta}}
\newcommand {\Cijab}{C_{ij}^{\alpha\beta}}
\newcommand {\Hijab}{H_{ij}^{\alpha\beta}}
\newcommand {\he}{h_{\mbox{ext}}}
\newcommand {\xia}{x_{i,\alpha}}
\newcommand {\vxk}{\vec{x}_{\mbox{{\small known}}}}
\newcommand {\vxu}{\vec{x}_{\mbox{{\small unknown}}}}
\newcommand {\Pab}{P_{\alpha\beta}}
\newcommand {\dab}{\delta^{\alpha\beta}}
\newcommand {\dij}{\delta_{ij}}
\newcommand {\erf}{\mbox{Erf}}

\title{Hybrid Lindblad dynamics and Bayesian inference from stochastic processes}
\author{H.J. Kappen\\ Donders Institute, Radboud University\\
Heyendaalseweg 135, 6525 AJ Nijmegen, the Netherlands}
\begin{document}
\maketitle
\begin{abstract}

We develop a Bayesian formulation of diffusive quantum-classical dynamics by treating the wave function and classical variables as components of an ordinary stochastic process. The joint probability density $\cP(\psi,x,t)$ obeys a classical Fokker-Planck equation, while the quantum state appears as its second moment. Requiring autonomous second-moment evolution and trace preservation, together with a full-column-rank quantum-classical noise correlation matrix, forces the ensemble dynamics into hybrid Lindblad form. The same result also determines the general class of diffusive stochastic unravelings compatible with this hybrid evolution. Notably, this construction makes positivity immediate and does not take complete positivity or an underlying unitary quantum dilation as assumptions, but instead relies on the presence of a sufficiently coupled classical sector.

We formulate quantum measurement, filtering and smoothing of hybrid quantum-classical dynamics as standard Bayesian inference and show how to correlation between the quantum and classical noise fully explains the differences between quantum smoothing and smoothing in classical time series models. We recover the stochastic master equation from the Kushner-Stratonovich equation with correlated noise and show how the quantum effect operator is related to the Bayesian backward message through the adjoint dynamics of the linear unraveling. The Bayesian posterior also defines a smoothed density matrix and, more generally, a posterior distribution over latent quantum-classical trajectories.

These quantities can be approximated with standard particle filtering and smoothing methods. Numerical examples show that smoothing improves reconstruction of a hidden quantum-classical trajectory and that the full trajectory posterior can retain structure, such as multimodality, that is absent from its density-matrix second moment. 
\end{abstract}
%\tableofcontents
\section{Introduction}

The theoretical framework for how quantum systems evolve in interaction with an environment is  the theory of open quantum systems. At the Markovian level,  the quantum system dynamics is described by the so-called Gorini-Kossakowski-Sudarshan-Lindblad (GKSL) equation  which, in addition to the usual Hermitian terms, contains terms describing the effect of the interaction of the system with the environment. 

One standard route to open quantum-system dynamics starts from a system \(\cS\) embedded in a larger quantum system, consisting of \(\cS\) and its environment \(\cE\), whose joint dynamics is unitary. The reduced dynamics maps the system density matrix as \(\rho_\cS\mapsto\Phi_t(\rho_\cS)\) and must be linear, trace preserving, and positive, \(\rho_\cS\geq0\Rightarrow\Phi_t(\rho_\cS)\geq0\). These conditions concern the reduced system itself. To obtain the GKSL structure one imposes the stronger requirement of \emph{complete positivity} (CP): the system dynamics must remain positive when the system is embedded in an arbitrarily larger quantum system with which it may be entangled. Thus, for any ancillary quantum system \(\cA\),

$$
X\geq0\quad\Rightarrow\quad(\Phi_t\otimes I_\cA)(X)\geq0 .
$$
In this sense, complete positivity is not only a condition on the reduced system dynamics, but a consistency requirement on that dynamics when the system is regarded as part of any larger quantum system. We review these standard constructions in section~\ref{CP}.

In this paper, we {\em derive the same open quantum system equations and their unravelings without the need to assume complete positivity or the existence of a unitary environment}.  Instead, the starting point is to postulate a hybrid stochastic quantum classical dynamics of a quantum state $\psi$ and a classical variable $x$.  Its probability density $\cP(\psi,x,t)$ obeys a classical Fokker-Planck equation, and the joint quantum-classical density operator appears as the second moment
\[\rho_J(x,t)=\int d\psi \cP(\psi,x,t)\psi\psi^\dagger\]
The idea to consider a classical probability $\cP$ on Hilbert space is not new \cite{breuer2002theory}, but its implications for quantum state estimation have not been fully explored. 
Theorem~\ref{thm4} shows that requiring this second moment to evolve autonomously (the dynamics is a function of the density matrix only, not requiring higher moments), together with trace preservation and a condition how the classical and quantum system interact (the covariance matrix $\Gamma$ of the classical and quantum noise must have full column rank), forces the ensemble evolution into the usual hybrid Lindblad form. Positivity of $\rho_J$
 is automatic because it is a second moment of an ordinary probability distribution.
In addition, the proof determines the most general form of the unravelings, with scalar functions \(c_a, \lambda\) parametrizing different unravelings of the same hybrid evolution. Thus we show that  the CP assumption is replaced by a condition on the column rank of $\Gamma$ and indeed the derivation fails in the absence of a classical sector.  

{\em  
A second contribution of the paper is to express existing quantum filtering and smoothing methods in a classical Bayesian framework}.
We give a graphical model depiction of this time series model in fig.~\ref{hmm}. The filtered distribution $
\cP(\psi_t\mid x_{0:t})$ with $t\le T$
describes what can be inferred causally from the measurement record. The filtered density matrix is its second moment
\bea
\rho_t^F=\int d\psi \cP(\psi_t\mid x_{0:t}) \psi\psi^\dagger\label{rhoF}
\eea
and its evolution follows from the classical Kushner-Stratonovich equation. We give the derivation and review the literature in section~\ref{quantumfiltering}. 

The Bayesian construction also clarifies what is meant by quantum smoothing. Due to the correlation between the process noise $d\xi$ and the observation noise $dW$ it has a different structure than the traditional classical hidden Markov model but reduces to it when $d\xi$ and $dW$ are uncorrelated. This provides a classical probabilistic explanation of the distinctive structure of quantum fsmoothing, without requiring operator noncommutativity as an explanatory ingredient. We distinguish between smoothing estimates of classical variables (retrodiction) and the smoothed density matrix and we review previous results. From the structure of the Bayesian model it is immediately clear that retrodiction can not depend on the smoothed density matrix and requires two objects:  the filtered density matrix and the effect operator, respectively, jointly known as the past quantum state \cite{tsang_time-symmetric_2009,gammelmark_past_2013}. Secondly, one can define the smoothed density matrix
\bea
\rho_t^S =\int d\psi \cP(\psi_t\mid x_{0:T}) \psi\psi^\dagger\label{rhoS1}
\eea
The smoothed density matrix should give a more accurate estimate of the quantum state because it is based on past and future date, instead of on past data alone. Eq.~\ref{rhoS1} generalizes previous proposals for Gaussian quantum models \cite{zhang_estimating_2020} and the proposal of \cite{guevara_quantum_2015}, as we discuss in section~\ref{smoothing}.

The quantities $\rho_t^F,\rho_t^S$ and $\cP(\psi_{0:T}\mid x_{0:T})$ are intractable to compute exactly. But since the latent dynamics is an ordinary stochastic process, standard sequential Monte Carlo methods, in particular particle filtering methods, can be used without constructing a separate quantum smoothing algorithm (section~\ref{section:filtering}).
This becomes particularly useful away from the saturated-noise limit, where the observed record does not determine a unique pure-state trajectory.
Particle filtering for (filtered) quantum state estimation was considered by \cite{ralph_multi-parameter_2017}, using a hybrid approach that combines the
stochastic master equation Eq.~\ref{eq:rho-filter} with sequential Monte Carlo methods to estimate several
unknown, but fixed, Hamiltonian parameters simultaneously with the conditioned quantum state.  In contrast, the particle construction considered in the present paper, and illustrated in the examples in section~\ref{numerical},  samples the entire unobserved dynamical hybrid state (quantum wave function, unobserved environment states coupled incoherently to the system, unobserved signals influencing the system coherently) and therefore provides posterior distributions over latent quantum-classical trajectories, making trajectory smoothing and path-dependent posterior quantities directly accessible.

We illustrate two consequences numerically. In the first example a hidden telegraph process modulates a continuously monitored qubit. {\em Smoothing with the complete observation record improves reconstruction of both the hidden telegraph process and the associated quantum state.} In the second example an unobserved environmental channel produces a posterior over pure-state trajectories with two well-separated modes even though their density-matrix average is close to maximally mixed. This illustrates a distinction that is central to the Bayesian formulation: {\em the density matrix is a sufficient object for operational predictions of measurement outcomes, but does not contain the full inferential structure of a posterior over the latent state $\psi_{0:T}$ in a specified unraveling.}

The main contributions of the paper are therefore threefold. First, we present a novel derivation of the hybrid Lindblad dynamics and its unravelings from an ordinary stochastic process, without the need to invoke the assumption of complete positivity or unitary dilation. Second, we formulate filtering and smoothing of hybrid quantum-classical dynamics as standard Bayesian inference (fig.~\ref{hmm}) and show how to correlation between the quantum and classical noise fully explains the differences between quantum smoothing and smoothing in classical time series models. Third, we obtain posterior distributions over latent quantum-classical states and trajectories, together with practical particle filtering methods for estimating them. We conclude that the classical Bayesian framework provides an elegant and sufficient (!) understanding of hybrid quantum-classical systems.

\section{A Bayesian interpretation of quantum-classical hybrid dynamics}
\subsection{Derivation of hybrid Lindblad equation and unravelings from Fokker-Planck equation}\label{fp2lindblad}
We model the open quantum system as a stochastic variable $\psi$ (the wave function in a given basis)  in interacting with a classical variable $x$. We assume a general diffusive dynamics for these two variables, given by the generic stochastic differential equation (SDE)
\bea
d\psi &=& f(x,\psi,t) \psi dt + g_a(x,\psi,t) \psi d\xi_a \nonumber\\
dx&=& h(x,\psi,t)dt + dW \label{fp2lindblad1}
\eea
with $\psi\in \C^d, x\in \R^n$ and $f,g_a\in \C^{d\times d}, a=1,\ldots m$,  $h \in \R^n$ and where throughout the paper repeated indices are summed over. $dW_i$ is real-valued and $d\xi_a$ is complex valued 
Wiener noise  with covariance matrices
\bea
\av{d\xi_a d\xi_b^*}=Q_{ab}dt\qquad \av{dW_i dW_{i'}}=C_{ii'}dt\qquad \av{dW_i d\xi_a}=\Gamma_{ia}dt\label{fp2lindblad1b}
\eea
The matrices $Q, C,\Gamma$ may depend on $\psi,x,t$.  

Since Eqs.~\ref{fp2lindblad1} describe a classical stochastic process, the marginal probability density $\cP(\psi,x,t)$ to observe the state $\psi,x$ at time $t$, given that the initial state $\psi_0,x_0$ at time $t_0$ (whose dependence we suppress in the notation), satisfies a Fokker-Planck (FP) equation $\dot \cP = \cL_{\rm F} \cP$ with $\cL_{\rm F} \cP$ given by Eq.~\ref{eq:complex_fokker_planck} below.  We define the joint density matrix  \footnote{We define 
$\int d\psi=\int \prod_{j=1}^d d\operatorname{Re}\psi_j\,d\operatorname{Im}\psi_j$ and $\av{}$ denotes integration with respect to $\cP(\psi,x,t)$ over $\psi$.} 
\bea
\rho_J(x,t) = \int d\psi \cP(\psi,x,t) \psi\psi^\dagger=\av{\psi\psi^\dagger}\label{joint}
\eea
as the second moment of $\cP(\psi,x,t)$. 
We refer to Eqs.~(4)--(6), describing the joint stochastic evolution of the quantum state $\psi$  and the classical variable $x$, together with the induced operator-valued state $\rho_J(x,t)$, as a \emph{hybrid quantum--classical dynamics}.

The core postulate of the hybrid quantum-classical dynamics is that $\rho_J(x,t)$ provides a sufficient description of the quantum system in interaction with measurement outcome. $\rho_J(x,t)$ reduces to the marginal density matrix $\rho(t)$ by integrating over $x$:
\bea
\rho(t) = \int dx \rho_J(x,t) =\int d\psi \cP(\psi,t) \psi\psi^\dagger\label{rho}
\eea
with $\cP(\psi,t)=\int dx \cP(\psi,x,t)$. Given the dynamics Eqs.~\ref{fp2lindblad1} and the definition Eq.~\ref{joint} we derive the following result.

%------ chatgpt
\begin{lemm}
\label{lem:rhoJ_dynamics}
$\rho_J(x,t)$ satisfies
\begin{align}
    \frac{\partial\rho_J}{\partial t}
    ={}&
    \left\langle
        f\psi\psi^\dagger
        +
        \mathrm{h.c.}
        +
        Q_{ab}
        g_a\psi\psi^\dagger g_b^\dagger
    \right\rangle -
    \partial_i
    \left\langle
        \Gamma_{ia}g_a\psi\psi^\dagger
        +
        \mathrm{h.c.}+
        h_i\psi\psi^\dagger
    \right\rangle
+
    \frac{1}{2}
    \partial_i\partial_{i'}
    \left\langle
        C_{ii'}\psi\psi^\dagger
    \right\rangle,
    \label{fp2lindblad2}
\end{align}
where $\partial_i=\partial/\partial x_i$.
\end{lemm}

\begin{proof}
See Appendix~\ref{prooflemma1}. 
\end{proof}

Note that the dynamical Eq.~\ref{fp2lindblad2} is not 'autonomous', i.e.  a function of $\rho_J$ only, because $f,g_a$ may depend on $\psi$ so that the expectations involve moments higher than $\av{\psi\psi^\dagger}$. But we can show that when we demand that the dynamics of $\rho_J$ is autonomous in $\rho_J$ its form is unique.

\begin{thrm}
\label{thm4}
Assume that the covariance matrices $C$, $Q$ and $\Gamma$ are
independent of $\psi$ and $\Gamma$ has full column rank and the marginal density matrix Eq.~\ref{rho} is trace preserving. 
Demanding that the evolution equation for $\rho_J(x,t)$ Eq.~\ref{fp2lindblad2} is autonomous in $\rho_J$ implies that it is linear and of the hybrid Lindblad form
\bea
\dot \rho_J&=& \cL_L (\rho_J) -\partial_i \left( f_i^c\rho_J +[\Gamma_{ia} L_a\rho_J  +h.c.]\right) + \frac{1}{2} \partial_i\partial_{i'}\left( C_{ii'}\rho_J\right)\label{fp2lindblad3}
\eea 
where we define the Lindblad super operator
\bea
\cL_L (X) &=&-i[H,X] +Q_{ab} \left(L_a X L_b^\dagger -\frac{1}{2}\{L_b^\dagger L_a, X\}\right) \label{lindblad}
\eea
with $H(x,t)$ a linear Hermitian operator, $L_a(x,t), a=1,\ldots,m$ linear operators.

We refer to Eq.~\ref{fp2lindblad3} as the hybrid dynamics.
The unraveling realizing this equation is of the form
\bea
f=-i H-\frac{1}{2}Q_{ab} \left(L_b^\dagger L_a  - 2 L_a c_b^*+c_a c_b^*I\right)+i\lambda(\psi) I\qquad g_a= L_a -c_aI\qquad h_i=f^c_i +[\Gamma_{ia}c_a+h.c.]\label{fp2lindblad4a}
\eea
with $f^c(x,t)$ a real function that we call the classical force, $\lambda(\psi,x,t)\in \R$ an arbitrary function, $I$ is the identity operator, and $c_a(\psi,x,t) \in \C$ arbitrary functions. 
\end{thrm}

\begin{coro}
\item Since \(Q\) is a covariance matrix, \(Q\ge0\). A unitary transformation in noise-channel space diagonalizes \(Q\), together with the corresponding transformation of \(d\xi_a\) and \(L_a\), and brings Eq.~\ref{lindblad} to
\[
\mathcal L_L(\rho)
=-i[H,\rho]
+\sum_a q_a\left(
\widetilde L_a\rho\widetilde L_a^\dagger
-\frac12\{\widetilde L_a^\dagger\widetilde L_a,\rho\}
\right),
\qquad q_a\ge0.
\]
This is the canonical GKSL form. Hence, although complete positivity was not assumed in Theorem 2, the derived marginal quantum dynamics is completely positive \cite{gorini1976completely,lindblad1976generators}.
\end{coro}
\begin{proof}
Autonomy of Eq.~\ref{fp2lindblad2} implies autonomy of the three terms that are zeroth, first and second order in $\partial_i$. This follows because 
autonomy is required for arbitrary \(P(\psi,x)\), the values and spatial derivatives of its conditional \(\psi\)-distribution can be varied independently at a fixed \(x\). Hence terms of different differential order in \(x\) must close separately.
Since $C$ does not depend on $\psi$, the second order derivative term in Eq.~\ref{fp2lindblad2} is already autonomous and linear in $\rho$. 

\subsubsection*{The first order term}
We consider the first order derivative term in Eq.~\ref{fp2lindblad2}
\bea
J_i(\psi) = \Gamma_{ia}g_a(\psi)P_\psi + \Gamma_{ia}^* P_\psi g_a^\dagger(\psi)+h_i(\psi) P_\psi\label{current}
\eea
with $P_\psi=\psi\psi^\dagger$.   In the following we suppress the $x,t$ dependence. $J_i(\psi)$ is a non-linear function of $\psi$. We demand autonomy, which  implies that there exists a map \(\mathcal J_i\) such that
$$ \langle J_i(\psi)\rangle=\mathcal J_i(\rho), \qquad \rho=\langle P_\psi\rangle . $$
For two distributions \(P_1,P_2\), with second moments \(\rho_1,\rho_2\), the mixture \(P=aP_1+(1-a)P_2\), \(0\le a\le1\), gives

$$ \mathcal J_i\!\left(a\rho_1+(1-a)\rho_2\right) = a\mathcal J_i(\rho_1)+(1-a)\mathcal J_i(\rho_2). $$
Thus \(\mathcal J_i\) is affine on the set of attainable second moments. Since \(\psi\) is not restricted to unit norm, \(P_{\sqrt{s}\psi}=sP_\psi\) for \(s\ge0\), so that \(\mathcal J_i(sP_\psi)=s\mathcal J_i(P_\psi)\). Hence \(\mathcal J_i\) is linear on the positive cone, and therefore extends uniquely to a real-linear Hermiticity-preserving map on Hermitian matrices and to a complex-linear map on \(M_d(\mathbb C)\). For \(P(\phi)=\delta(\phi-\psi)\),

$$ J_i(\psi)=\mathcal J_i(P_\psi). $$
Thus the nonlinear map \(\psi\mapsto J_i(\psi)\) Eq.~\ref{current} is represented by a linear map \(P_\psi\mapsto J_i(P_\psi)\).

Define the vector field $w_i(\psi)=\Gamma_{ia}g_a(\psi)\psi$ and $v_i(\psi)=w_i(\psi)+\frac12 h_i(\psi) \psi$. Then Eq.~\ref{current} is written as $J_i(\psi) = v_i (\psi)\psi^\dagger + \psi v_i(\psi)^\dagger$. Therefore $J_i(\psi)=\cJ_i(P_\psi)$ is tangent to the rank-one cone at $P_\psi$
\[
\Pi^\perp_{\psi} \cJ_i(P_\psi) \Pi^\perp_{\psi}=0
\]
with $\Pi^\perp_\psi = I - \frac{\psi\psi^\dagger}{\|\psi\|^2}$. 
One can then show that 
\begin{lemm}
\label{lemma3}
Let \(\cR:M_d(\mathbb C)\to M_d(\mathbb C)\) be a complex-linear, Hermiticity-preserving map. Suppose that for every nonzero \(\psi\),
\bea
\Pi_\psi^\perp\, \mathcal R(P_\psi)\, \Pi_\psi^\perp=0, \qquad P_\psi=\psi\psi^\dagger\label{tangency}
\eea
where \(\Pi_\psi^\perp\) projects onto the subspace orthogonal to \(\psi\). Then there exists a fixed operator \(B\) such that
\[
\mathcal R(X)=BX+XB^\dagger
\]
for every \(X\). 
\end{lemm}
\begin{proof}
See Appendix~\ref{appendix:lemma3}.\end{proof}
The intuition of the Lemma is that the tangency Eq.~\ref{tangency} means that for every pure state \(P_\psi\),
$$ \mathcal J(P_\psi)=v(\psi)\psi^\dagger+\psi v(\psi)^\dagger=B(\psi) P_\psi + P_\psi B(\psi)^\dagger. $$
where we defined $B(\psi)$ such that $B(\psi)\psi =v(\psi)$. The nontrivial statement of Lemma~\ref{lemma3} is that
the linearity of \(\mathcal J\) constrains the tangent field \(v(\psi)\) to be generated by a single linear operator \(B\), up to an arbitrary imaginary component parallel to \(\psi\), which is invisible in \(P_\psi=\psi\psi^\dagger\).
\footnote{One can write $v(\psi) = B\psi + i\lambda(\psi) \psi$ with $\lambda(\psi)\in \R$, because the extra term does not change the tangent matrix: $(i\lambda \psi)\psi^\dagger + \psi (i\lambda\psi)^\dagger =0$}
%\begin{proof}
%Define the cone of all positive rank-one operators $P_\psi=\psi\psi^\dagger$. Define the linear differential equation $\dot X_t = \cJ(X_t), X_0=P$. Define $\Phi_t (P):=X_t = e^{t\cJ}P$. If $\cJ(P)$ is tangent to $P$ at all $P$, then locally $\Phi_t$ stays on the manifold. 
%\end{proof}
Thus, using Lemma~\ref{lemma3},  Eq.~\ref{current} can be written as
\bea
J_i(\psi) = B_i P_\psi + P_\psi B_i^\dagger\label{current1}
\eea
with some linear operator $B_i$. 

Applying $\Pi_\psi^\perp$ to $J_i$ yields $\Pi^\perp_\psi J_i = \Pi_\psi^\perp w_i \psi^\dagger = \Pi_\psi^\perp B_i P_\psi$. Therefore $\Pi_\psi^\perp w_i = \Pi_\psi^\perp  B_i \psi$. This defines the component of $w_i$ orthogonal to $\psi$. Therefore 
\[
w_i(\psi) =\Gamma_{ia}g_a(\psi) \psi = B_i \psi + \beta_i(\psi) \psi
\]

We now use the fact that $\Gamma$ has full column rank. This means that there exist a left inverse $R$ with components $R_{ai}$  such that $R_{ai}\Gamma_{ib}=\delta_{ab}$. Then 
\bea
g_a(\psi)\psi = R_{ai}B_i \psi +R_{ai}\beta_i(\psi) \psi=\left(L_a-c_a(\psi)I\right)\psi\label{lemma3ga}
\eea
where we define the linear operator $L_a = R_{ai}B_i$ and the non-linear scalar function $-c_a(\psi) = R_{ai}\beta_i(\psi)$.
\footnote{Strictly speaking $g_a =L_a -c_a I +K_a(\psi)$ with $K_a(\psi)\psi=0$. However, since $g_a$ appears only in the combination $g_a\psi$ we can ignore $K_a$.} Thus $g_a$ is exactly of the form Eq.~\ref{fp2lindblad4a}. Substituting $g_a$ in Eq.~\ref{current} yields
\[
J_i(\psi) = \Gamma_{ia} L_a P_\psi + \Gamma_{ia}^*P_\psi L_a^\dagger +(h_i -\Gamma_{ia}c_a -\Gamma_{ia}^*c_a^*)P_\psi
\]
The first two terms are already linear in $P_\psi$. The last term is linear in $P_\psi$ provided that the coefficient between brackets is independent of $\psi$:
\[
h_i -\Gamma_{ia}c_a -\Gamma_{ia}^*c_a^* =f_i^c(x,t)
\]
which coincides with Eq.~\ref{fp2lindblad4a}.

\subsubsection*{The zeroth order term}
We now consider the zeroth order derivative term of Eq.~\ref{fp2lindblad2}. We use Eq.~\ref{lemma3ga} and write
\beaa
\cG &=& f\psi\psi^\dagger + \mathrm{h.c.} +Q_{ab}g_a\psi\psi^\dagger g_b^\dagger=b\psi^\dagger + \psi b^\dagger+Q_{ab}(L_a -c_a) P_\psi (L_b -c_b)
\\&=& Q_{ab} L_a P_\psi L_b^\dagger + \tilde b \psi^\dagger + \psi \tilde b^\dagger
\eeaa
where $b(\psi) =f(\psi)\psi$ and $\tilde b = b -Q_{ab}c_b^* L_a \psi +\frac12 Q_{ab} c_a c_b^* \psi$. 
We impose linear autonomous closure again. The first term in $\cG$ is already linear. Therefore the remainder $\cR(P_\psi) = \tilde b \psi^\dagger + \psi \tilde b^\dagger$ must itself be linear. Since $\cR(P_\psi)$ is Hermitian and satisfies the tangent condition Eq.~\ref{tangency}, Lemma~\ref{lemma3} applies again and 
\bea
\cR(P_\psi) =\tilde b \psi^\dagger + \psi \tilde b^\dagger = AP_\psi +P_\psi A^\dagger\label{lemma3btilde}
\eea
for some fixed operator $A$. 
Thus 
\[
\cG(P_\psi) = Q_{ab} L_a P_\psi L_b^\dagger +  AP_\psi +P_\psi A^\dagger
\]
Since $\dot \rho = \cG(\rho)$ describes the evolution of the open quantum system in the absence of measurement outcomes, $\rho$ must remain normalized. Therefore $\Tr(\cG(\rho))=0$ which implies
\[
A +A^\dagger =-Q_{ab}L_b^\dagger L_a
\]
which constrains the Hermitian part of $A$. The anti-Hermitian part is arbitrary so we can write $A=-iH -\frac12 Q_{ab}L_b^\dagger L_a$ with $H$ Hermitian, so that 
\beaa
\cG(\rho) = -i[H,\rho]+Q_{ab} \left(L_a \rho L_b^\dagger -\frac12\{L_b^\dagger L_a, \rho\}\right)
\eeaa
We therefore establish that the zeroth order derivative term in Eq.~\ref{fp2lindblad2} recovers the Lindblad equation. 

Using a similar reasoning as under Eq.~\ref{current1}, Eq.~\ref{lemma3btilde} implies that $\tilde b = A\psi + i\lambda(\psi) \psi$. $\tilde b$ satisfies Eq.~\ref{lemma3btilde} when $\lambda(\psi)\in \R$  and therefore
\[
b(\psi) =  f(\psi)\psi =\left(-iH  -\frac12 Q_{ab}\left(L_b^\dagger L_a -2c_b^* L_a -c_a c_b^*\right)+i\lambda(\psi)I\right) \psi
\]
Therefore $f(\psi)$ is of the form Eq.~\ref{fp2lindblad4a}. 

\end{proof}

The stochastic dynamics Eqs.~\ref{fp2lindblad1}  become
\bea
d\psi&=&-i H \psi dt -\frac{1}{2}Q_{ab} \left(L_b^\dagger L_a  - 2c_b^* L_a +c_a c_b^*I\right)\psi dt +i\lambda(\psi)\psi dt+ (L_a -c_aI)\psi d\xi_a\label{sse3}\\
dx_i &=&f^c_idt + [ \Gamma_{ia}c_a+h.c.]dt+dW_i\label{dx}
\eea
We refer to Eq.~\ref{sse3} as  the stochastic Schr\"odinger equation (SSE). 
The hybrid dynamics Eq.~\eqref{fp2lindblad3}, together with the
unravelings Eqs.~\eqref{sse3} and~\eqref{dx}, gives the most general
diffusive linear quantum--classical dynamics. These equations 
agree with well-known results for diagonal $Q$ \cite{gisin1984quantum,percival1998quantum,doherty_quantum_2000,semina2014stochastic,barchielli2009quantum,wiseman2009quantum} as well as more recent results \cite{layton2024healthier,diosi_hybrid_2023}.
%\footnote{Eq.~\ref{fp2lindblad3} agrees with \cite{layton2024healthier} Eq. 14 and Eqs.~\ref{sse3} and~\ref{dx}  with  \cite{layton2024healthier}  Eqs.16 and 18 with $D_0^{\alpha\beta}\leftrightarrow Q_{ab}, D_{2,ij}^{00}\leftrightarrow C_{ij}, D_{1,i}^{0\alpha} \leftrightarrow \Gamma_{ia}, D_{1,i}^{00} \leftrightarrow f_i^c$ where \cite{layton2024healthier} restricts to the case where the quantum and classical noise are linearly related as $d\xi_a=\Gamma_{ia} C^{-1}_{ii'} dW_{i'}$ such that it saturates the noise condition Eq.~\ref{saturation}.}
%\footnote{Partial agreement. Eq.~\ref{fp2lindblad3} agrees with \cite{diosi_hybrid_2023} Eq. 36 with $D_0^{\alpha\beta}\leftrightarrow Q_{ab}, D_C^{nm}\leftrightarrow C_{ii'}, G_{CQ}^{na}\leftrightarrow \Gamma_{ia}, V^n \leftrightarrow f_i^c$. The unravelings Eqs.~\ref{sse3} and~\ref{dx}  do not
%agree with  \cite{diosi_hybrid_2023}  Eqs. 48-50 (NB the Erratum \cite{diosi2023erratum}). 
%In the case of a single Lindblad operator his Eq. 48 reduces to $-iH_\text{fr}=-\frac{1}{2}Q\left(L^\dagger L -2c L^\dagger +c^*c\right)$ which does not agree our result or with the other 
%diagonal results $-iH_\text{fr}=-\frac{1}{2}Q\left(L^\dagger L -2c^* L +c^*c\right)$.
%}
The marginal Lindblad equation $\dot\rho=\cL_L(\rho)$ is directly obtained from Eq.~\ref{rho} and integrating Eq.~\ref{fp2lindblad3} over $x$.

%\footnote{
%Using the definition $J_i(\psi) = w_i \psi^\dagger + \psi w_i^\dagger + h_i \psi\psi^\dagger$ and its derived form Eq.~\ref{current1} with the solution $w_i = B_i \psi +\beta_i\psi$ yields $\beta_i +\beta_i^* + h_i=0$. The relation $-c_a(\psi) = R_{ai}\beta_i(\psi)$ under Eq.~\ref{lemma3ga} does not generally imply $\beta_i(\psi) =-\Gamma_{ia} c_a$ because $\Gamma$ has maximal column rank so that $R\Gamma =I$ but not $\Gamma R =I$. Using Eq.~\ref{lemma3ga} we have 
%$
%w_i = \Gamma_{ia} L_a \psi -\Gamma_{ia} c_a \psi = B_i \psi +\beta_i\psi
%$ or 
%\[
%(B_i -\Gamma_{ia}L_a)\psi = -(\beta_i +\Gamma_{ia}c_a)\psi
%\]
%The operator $B_i -\Gamma_{ia}L_a$ maps every $\psi$ to a a vector proportional to $\psi$. Therefore it must be proportional to the identity: $B_i -\Gamma_{ia}L_a=\kappa_i I$. This implies $\beta_i =-\Gamma_{ia}c_a -\kappa_i$ with $\kappa_i$ independent of $\psi$. Consequently $h_i = -\beta_i -\beta_i^* = f_i^c +[\Gamma_{ia}c_a + h.c.]$ with $f_i^c = \kappa_i +\kappa_i^*$ the classical force. 
%
%This construction implies that when $\Gamma$ is invertible, one can absorb the classical force into the Lindblad operators.
%}
%

Thm.~\ref{thm4} is valid for any choice of the scalar $c_a$. The choice of $c_a$ affects the unraveling Eqs.~\ref{sse3} and~\ref{dx} but not the hybrid dynamics Eq.~\ref{fp2lindblad3}. 
In other words, different choices of $c_a$ define different unravelings. 
In order to specify $c_a$ we consider the evolution of the norm $\|\psi\|^2$. Define the stochastic variable $P_\psi=\psi\psi^\dagger$. Using It\^o calculus, the dynamics of $P_\psi$ is
%\footnote{
%\bea
%dP&=&d\psi \psi^\dagger +h.c. +d\psi d\psi^\dagger \nonumber\\
%&=& -i[H,P]dt -\frac{1}{2}Q_{ab}\{L_b^\dagger L_a,P\}dt +Q_{ab} (L_a c_b^* P + P L_b^\dagger c_a)dt -Q_{ab} c_a c_b^* Pdt +[(L_a -c_a) P d\xi_a + h.c ]\nonumber\\
%&&+Q_{ab}(L_a P L_b^\dagger -c_a P L_b^\dagger -L_a P c_b^* +c_a c_b^* P)dt \nonumber\\
%&=& \cL_L(P) + [(L_a -c_a) P d\xi_a + h.c ]\label{sme3}
%\eea
%}
\bea
dP_\psi&=&d\psi \psi^\dagger +h.c. +d\psi d\psi^\dagger = \cL_L(P_\psi) dt+ [(L_a -c_a) P_\psi d\xi_a + h.c ]\label{sme3}
\eea
with $\cL_L$ given by Eq.~\ref{lindblad}.
Note that by taking the expectation value of Eq.~\ref{sme3} we directly obtain Eq.~\ref{lindblad}. 
When $\Tr (P_\psi)=1$, the change in the norm is given as
\bea
d\left(\|\psi\|^2 \right) =\Tr (dP_\psi)= \Tr[(L_a -c_a)P_\psi]d\xi_a + h.c.=\left(\psi^\dagger L_a \psi-c_a\right)d\xi_a +h.c.\label{dZ}
%\label{normalization}
\eea
In particular,
note that $d\Tr (\rho_J) =\av{\Tr(dP_\psi)}=0$, i.e. $\Tr (\rho_J)$ is preserved for any choice of $c_a$ because $\av{d\xi_a}=0$. 

We consider two choices. 
\ben
\item {\bf Non-linear unraveling}. We may demand that the norm of any stochastic trajectory is conserved $\Tr(dP)=0$ by choosing $c_a =\psi^\dagger L_a \psi$. Note that this choice makes Eq.~\ref{sse3} or Eq.~\ref{sme3} non-linear in $\psi$. 
\footnote{While the choice $c_a =\psi^\dagger L_a \psi$ is sufficient to preserve $\|\psi\|$, it is not always necessary. For instance, when $Q$ is a real rather than Hermitian, 
we can choose $d\xi_a$ real. In this case the choice $c_a=\frac{1}{2}\psi^\dagger\left(L_a +L_a^\dagger\right)\psi$ also preserves the norm. In particular, when $L_a$ is anti-Hermitian, the unravelings are linear and norm preserving. This option was exploited in \cite{villanueva2024a}.}
\item {\bf Linear unraveling}. We may set $c_a=0$ in which case the unravelings Eq.~\ref{sse3} or Eq.~\ref{sme3} become linear in $\psi$. In this case the norm is not preserved for individual trajectories, $\Tr(dP_\psi)\ne 0$, but for the average $\av{d\Tr\left(P_\psi\right)}=0$. Since $c_a=0$, Eq.~\ref{dx} may suggest that the quantum dynamics does not depend on the classical state $x$. However, note that usually $\Gamma_{ia}=\av{dW_i d\xi_a}\ne 0$, so that the quantum and classical dynamics co-evolve through the correlated diffusion. 
\een
To cover both the linear and non-linear unravelings we define $c_a =\eta \psi^\dagger L_a \psi$ with $\eta=0,1$, respectively.

Comments:
\ben
\item The definition Eq.~\ref{joint} automatically ensures positivity of $\rho_J$. This is a necessary consequence because our starting point is a stochastic dynamics. In other derivations that are based on completely positive maps and Kraus operators the positivity of $\rho_J$ needs to be derived separately. See  \cite{diosi_hybrid_2023} and references there.
\item The linear (\(\eta=0\)) and norm-preserving (\(\eta=1\)) unravelings are related by a Doob transform, which we show in Appendix~\ref{sec:retrograde-filter}. For the linear unraveling, the stochastic norm
$
Z_t=\|\psi_t\|^2
$
is a positive martingale (see Eq.~\ref{dZ}), and it provides the trajectory weight that converts the linear path measure into that of the normalized unraveling. Equivalently, the normalized process is obtained by reweighting linear trajectories by their norm and subsequently normalizing the wave function. We derive the corresponding relation between the transition kernels \(P^{(0)}\) and \(P^{(1)}\) in Appendix~\ref{sec:retrograde-filter}, where it is used to connect the Bayesian backward message to the quantum effect operator.

\item
 If \(\lambda=\lambda(x,t)\) is independent of \(\psi\), it can be absorbed into the Hamiltonian by the shift \(H\to H-\lambda I\). More generally, a nonzero \(\lambda(\psi,x,t)\) defines a valid unraveling that leaves the hybrid dynamics, Eq.~\ref{fp2lindblad3}, unchanged. We do not pursue this additional freedom here and set \(\lambda=0\) in the remainder of the paper.
\item Let \(\Gamma\) be an \(n\times m\) matrix, where \(n\) is the dimension of the classical variable \(x\) and \(m\) is the number of quantum noise components \(d\xi_a\). The full-column-rank assumption on \(\Gamma\) in Thm.~\ref{thm4} requires its \(m\) columns to be linearly independent, and therefore \(n\geq m\geq 1\). In particular, the theorem does not apply when \(\Gamma=0\), i.e. when the classical and quantum sectors are completely decoupled.

Thus, the theorem assumes that the \(m\) quantum noise channels are coupled to at least \(m\) classical degrees of freedom. These classical variables need not be observed. If they are marginalized out, one recovers the usual Lindblad equation,

$$
\dot\rho=\mathcal L_L(\rho).
$$

Hence the standard open-system Lindblad dynamics follows from the present construction provided that the quantum system is embedded in, and coupled to, a larger quantum--classical stochastic dynamics.

Without this coupling assumption, Thm.~\ref{thm4} does not exclude more general positive open-system dynamics. This is analogous in spirit to the standard derivation of Lindblad dynamics from complete positivity: complete positivity constrains the evolution by requiring positivity to be preserved when the system is extended by an arbitrary ancilla, whereas positivity of the reduced system alone is insufficient to enforce Lindblad form.

\item
The full-column-rank assumption  is independent of the saturation condition \(Q=\Gamma^\dagger C^{-1}\Gamma\) that we discuss in section~\ref{section:noise_covariance}. In particular, \(\Gamma\) may have full column rank while the noise remains non-saturated, so that a residual quantum-noise component is not determined by the classical noise.
\een

\subsection{The noise covariance matrix}
\label{section:noise_covariance}
The covariance matrix of the unravelings Eqs.~\ref{sse3} and~\ref{dx} is of dimension $n+ m$ and of the form
\bea
D=\left(\begin{tabular}{cc} $ C$ & $\Gamma$\\ $\Gamma^\dagger$ & $Q$ \end{tabular} \right)\label{D}
\eea
%\footnote{
%Given a mean zero Gaussian distribution $p(x_a,x_b)$ with covariance matrix 
%\beaa
%\Sigma = \left(\begin{tabular}{cc}$\Sigma_{aa}$ & $\Sigma_{ab}$\\$\Sigma_{ba}$ & $\Sigma_{bb}$\end{tabular}\right)
%\eeaa
%The conditional $p(x_a|x_b)$ is mean zero Gaussian with covariance matrix 
%$\Sigma_{a|b}=\Sigma_{aa} - \Sigma_{ab} \Sigma_{bb}^{-1}\Sigma_{ba}$
%
The derivations in \cite{diosi_hybrid_2023,oppenheim2022two,oppenheim2023gravitationally,layton2024healthier} derive the generalized Lindblad equation Eq.~\ref{fp2lindblad3} from complete positivity and imply positive-semidefinite noise matrices $Q,C\ge 0$, together with positivity of the joint covariance matrix 
$D\ge 0$. In the present work we restritct to the nondegenerate case $Q,C>0$ so that the Schur complement condition can be written as
\bea
Q \ge \Gamma^\dagger C^{-1} \Gamma\qquad C \ge \Gamma Q^{-1} \Gamma^\dagger\label{noise_bound}
\eea
Our derivation starts instead from the unraveling dynamics, for which $D\ge 0$ follows directly from the definition of the joint noise covariance in Eq.~\ref{fp2lindblad1b}. The assumption $Q,C>0$ is made only to allow use of ordinary inverses throughout. 
For a given coupling $\Gamma$ between the quantum and classical system, the lower bound $Q_\text{min} = \Gamma^\dagger C^{-1} \Gamma$ shows that the minimum decoherence increases as the classical diffusion decreases, and vice versa.

The bound $Q=\Gamma^\dagger C^{-1} \Gamma$ is saturated (the so-called saturated noise case) when the quantum and classical noise are maximally correlated, i.e. $d\xi = F dW$ with $F$ a (complex) $m \times n$ matrix. In this case the quantum increment, and therefore the quantum state, is fully determined by the classical increment $dx$.
In the non-saturated noise case, $\psi(t)$ is a stochastic quantity given by the filtered estimate $\cP(\psi|x_{0:t})$ that we consider in section~\ref{quantumfiltering}. 

Although $d\xi_a \in \C$, we argue in Appendix~\ref{noisechoice} that we can choose $d\xi_a\in \R$ without loss of generality. 

For a single channel, $C,Q$, and $\Gamma$ as defined in Eq.~\ref{fp2lindblad1b} are real scalars. It is customary in the quantum-filtering literature \cite{wiseman1993quantum,doherty_quantum_2000}
 to express their relative strength in terms of the detector efficiency 
\beaa
\eta_{\text{eff}} = \frac{\Gamma^2}{CQ}
\eeaa
$0\le \eta_\text{eff} \le 1$ which follows from the positivity condition Eq.~\ref{noise_bound}. 
When $
C = Q = 1$ we obtain $
\Gamma = \sqrt{\eta_{\rm eff}}$. 
The decomposition of the quantum noise $d\xi$ into a component correlated with the
observation noise $dW$ and an independent residual component then becomes
\begin{equation}
d\xi_t
=
\sqrt{\eta_{\rm eff}}\,dW_t
+
\sqrt{1-\eta_{\rm eff}}\,dV_t,
\end{equation}
where $dV_t$ is a Wiener increment independent of $dW_t$, $\av{dV_t dW_t}=0$. Thus
$\eta_{\rm eff}$ measures the fraction of the quantum-noise variance that is
correlated with the observation noise, while $1-\eta_{\rm eff}$ is the fraction
that remains unobserved.
For $\eta_{\rm eff}<1$,
the observed trajectory $x_{0:T}$ does not uniquely determine the
underlying pure-state trajectory. The component of $d\xi_t$ proportional to
$dV_t$ remains latent even when the complete observation record is known. As we show the numerical examples, 
the posterior distribution over trajectories can then 
be approximated by particle smoothing.

\section{Relation to previous work}
\label{relatedwork}
In section~\ref{CP} we review the standard literature on the GKSL equation and its unraveling with a focus on the assumption of complete positivity or unitary dilation. In section~\ref{quantumfiltering} we review previous results on quantum filtering and derive the stochastic master equation for the conditioned state evolution using the Bayesian framework.  In section~\ref{smoothing} we review some earlier approaches to retrograde prediction and smoothing estimates of the quantum state and how these are understood in the Bayesian framework of fig.~\ref{hmm}.  
\subsection{Markovian open quantum systems}
 \label{CP}

Standard approaches to derive the GKSL equations distinguish between the ensemble and trajectory levels. Here we briefly review these theories. At the ensemble level, the evolution of the density operator follows from requiring a continuous semigroup of completely positive trace-preserving maps  \cite{gorini1976completely, lindblad1976generators,Kraus1983}. 
Alternatively, the GKSL equation can be derived from an underlying unitary dynamics of a quantum system coupled to a quantum environment, followed by tracing out the environment and taking an appropriate weak-coupling Markovian limit \cite{Davies1974,davies1976quantum}. In this construction complete positivity is not imposed directly on the reduced dynamics, but follows from its origin as the reduced dynamics of the larger unitary quantum system.

At the trajectory level, the evolution of the density matrix according to the GKSL equation can be represented by stochastic evolutions of pure states, or unravelings \cite{carmichael1993open}. 
The stochastic dynamics describes the evolution of the quantum state as a stochastic variable, possibly together with the evolution of classical measurement outcomes \cite{davies1969quantum,gisin1984quantum,belavkin1983theory}. \cite{wiseman2001complete} characterized the general diffusive pure-state unravelings of a given Markovian Lindblad master equation; CP is therefore already encoded in the ensemble dynamics used as the starting point.
\cite{oppenheim2022two, diosi_hybrid_2023,layton2024healthier} generalize this construction to the hybrid quantum-classical dynamics. Alternatively, the stochastic unraveling can be derived from a continuous-measurement model in which the system interacts unitarily with a quantum environment and the output field is continuously monitored \cite{bouten_introduction_2007}. The conditional stochastic equation and the unconditional GKSL master equation then follow from the same underlying dilation; complete positivity of the unconditional dynamics is guaranteed by the unitary system-environment evolution and subsequent partial trace.

A conceptually different route is provided by {\em objective collapse models}, also referred to as dynamical or spontaneous collapse models \cite{pearle1976reduction,ghirardi1986unified,diosi1989models}. In these theories, the stochastic evolution of the quantum state is not interpreted as an unraveling of a given ensemble dynamics with no physical meaning. Rather, the stochasticity is postulated as ontological, ie. a fundamental aspect of physical reality. For Markovian collapse models, averaging over the stochastic noise typically leads to a master equation of GKSL form. 
The early models of Di\'osi \cite{diosi1987,diosi1989models}, later closely associated with Penrose's proposal for gravity-induced state reduction \cite{penrose1996}, have been reformulated as a hybrid quantum-classical dynamics by \cite{tilloy2016,tilloy2017}, where the quantum mass density is continuously monitored and the resulting classical signal corresponds to the gravitational field, which in turn acts back on the quantum system. 
This idea has recently received renewed attention 
 \cite{oppenheim2023gravitationally}.

\subsection{Quantum filtering}
\label{quantumfiltering}
Quantum filtering is the estimation of the conditional quantum state from the past record of a continuous measurement \cite{wiseman2009quantum,barchielli2009quantum,jacobs2014quantum}. In quantum sensing, the same measurement record can also be used to infer unknown classical parameters or time-dependent signals that affect the quantum dynamics. A number of works have developed joint state-and-parameter or state-and-signal estimation schemes of this kind, combining the stochastic master equation for the conditioned quantum state with Bayesian or classical filtering of the unknown quantity \cite{ralph_frequency_2011,gambetta_state_2001,verstraete_sensitivity_2001,stockton_robust_2004,negretti_estimation_2013,chase_magnetometry_2009}.

The filtered density matrix $\rho_t^F$ Eq.~\ref{rhoF} satisfies a dynamical equation, known as the {\em stochastic master equation}.  We show how  this equation results from our Bayesian framework using It\^o calculus. We first derive the Kushner-Stratonovich (KS) equation for the filtered distribution $\pi_t(\psi)=\cP(\psi\mid x_{0:t})$ in Appendix~\ref{section:KS}. In this derivation it is essential to include the correlation between the noise in the latent dynamics $d\xi$ and the observation noise $dW$, which is not standard in most time-series models. The result is Eq.~\ref{eq:KS-correlated}.

We now apply Eq.~\eqref{eq:KS-correlated} to the hybrid dynamics Eqs.~\ref{sse3} and~\ref{dx} where the hidden variable is the quantum state \(\psi\), while the
observed variable is the classical trajectory \(x_{0:t}\). 
The object of interest is not the full filtered distribution over quantum
states, but its second moment, i.e. $\rho_t=\pi_t(P_\psi)$ as given by Eq.~\ref{pi_definition} with $\phi=P_\psi=\psi\psi^\dagger$. The stochastic equation for \(P_\psi\) is given by Eq.~\ref{sme3}.
Note that $P_\psi$ is complex valued while the derivation above was done for real variables. This generalization is valid when realizing that the complex dynamics can always be written as real dynamics in double the dimension. The only term that needs attention is $\pi_t(dP_\psi dW_i)=\Gamma_{ia}(L_a-c_a)P_\psi dt+\mathrm{h.c.}$. 
 Eq~\ref{eq:KS-correlated} becomes
\begin{align}
\begin{aligned}
 d\rho_t
 =
 &\cL_{L}(\rho_t)dt
 +
 \left[
   \pi_t(P h_i)
   -\rho_t\pi_t(h_i)
   +
   \pi_t\left(
     \Gamma_{ia}(L_a-c_a)P+\mathrm{h.c.}
   \right)
 \right]
 C^{-1}_{ij}dI_j ,
\end{aligned}
\label{eq:rho-filter-general}
\end{align}
Substituting $h_i$ from Eq.~\ref{dx}, Eq.~\ref{eq:rho-filter-general} gives 
 \begin{align}
\boxed{
\begin{aligned}
 d\rho_t
 =
 &\cL_L(\rho_t)dt+
 \left[
   \Gamma_{ia}(L_a-\bar c_{a,t})\rho_t+\mathrm{h.c.}
 \right]
 C^{-1}_{ij}
 dI_{j,t}\\
 dI_{j,t}=&
   dx_{j,t}-
   \left(
     f_{j,t}^c+\Gamma_{jb}\bar c_{b,t}+\mathrm{h.c.}
   \right)dt
\end{aligned}
}
\label{eq:rho-filter}
\end{align}
where  $\bar c_{a,t}=\pi_t(c_a)=\eta \Tr(L_a \rho_t)$ with $\eta=0,1$ (the linear and non-linear unravelings) and $\cL_{L}$ is given by Eq.~\ref{lindblad}.
Eq.~\ref{eq:rho-filter} is autonomous, because it depends only
on \(\rho_t\) and the observed trajectory \(x_{0:t}\).
Eq.~\ref{eq:rho-filter} is the conditioned
counterpart of the marginal hybrid equation Eq.~\ref{fp2lindblad3}. The
spatial derivative terms are absent because the observed trajectory \(x_{0:t}\) is now
given.
Eq.~\ref{eq:rho-filter} agrees with the result of 
\cite{wiseman1993quantum} who first derived this result using Kraus operators.

Note that 
$\av{dI_i dI_j}=C_{ij}dt$ and $\av{dI_i d\xi_a} = \Gamma_{ia} dt$. 
 Thus $dI_t$ and $dW_t$ have the same quadratic variation, but they they are very different objects.
The increment \(dW_t\) is part of the underlying hybrid dynamics, whereas \(dI_t\) is defined relative to the observer's conditional prediction of the measurement record,
$$
dI_t = dx_t-\mathbb E[dx_t\mid x_{0:t}].
$$
Thus \(dI_t\) is not an independent physical noise introduced into the model, but the unpredictable part of the observed increment after conditioning on the past record.  The Wiener increment that appears in the usual stochastic-master-equation formulation is therefore the innovation process \(dI_t\), not the primitive noise \(dW_t\).
See also \cite{doherty_quantum_2000}.

The filtering equation Eq.~\ref{eq:rho-filter} describes the dynamics of the normalized ($\eta=1)$ and unnormalized ($\eta=0$) filtered density matrix, which we denote by $\rho_t$ and $\sigma_t$, respectively. In the norm-preserving case, Eq.~\ref{eq:rho-filter} is non-linear 
because \(\bar c_{a,t}\) depends on \(\rho_t\). 
In order to compute quantum expectations, our interest is to compute $\rho_t$.
However, we show that it can be computed from $\sigma_t$. Furthermore, $\sigma_t$ also contains information about the data likelihood. 
\begin{lemm}
\label{lemma:linear}
$\rho_t$ and $\sigma_t$ are related as 
\bea
\rho_t&:=&\frac{\sigma_t}{\Tr\sigma_t} \label{eq:rho-sigma-section}\\
\Tr(\sigma_t) &=&\Lambda_t := \frac{\cP^{(1)}(x_{0:t})}{\cP^{(0)}(x_{0:t})} \label{eq:TRsigma}
\eea
Here $\cP^{(0)}(x_{0:t})$ and $\cP^{(1)}(x_{0:t})$  denote the likelihood densities of the observed record under the linear and normalized unravelings, respectively. For continuous observation paths the individual likelihood densities depend on the common path-density convention, whereas their ratio does not.
\end{lemm}
\begin{proof} See Appendix~\ref{appendixlemma3}.\end{proof}

\subsection{Quantum smoothing}
\label{smoothing}
In this section we review and discuss quantum smoothing using the classical Bayesian framework of the unravelings. We use a
time discretization $dt$ and $t=kdt, k=0,1,2,\ldots$ and define $dx_t=x_{t+dt}-x_t$.
The unravelings Eqs.~\ref{sse3} and~\ref{dx} define a joint probability distribution on the quantum and classical trajectories $\cP(\psi_{0:T},x_{0:T})$ which 
is of the form
\bea
\cP(\psi_{0:T},x_{0:T})=\prod_t
\cP(\psi_{t+dt}\mid\psi_t,x_t)
\cP(dx_t\mid\psi_t,\psi_{t+dt},x_t).\label{markov}
\eea
Smoothing is defined as estimating the posterior $\cP(\psi_t\mid x_{0:T})$. Since $x_{0:T}$ is given and constant we focus on the graphical structure in $\psi_{0:T}$. It  
 is depicted in fig.~\ref{hmm}, where the circles are stochastic variables and the arrows depict conditional dependencies.
The term $\cP(\psi_{t+dt}\mid\psi_t,x_t)$ in Eq.~\ref{markov} describes the quantum dynamics Eq.~\ref{sse3} and is denoted by the horizontal arrows. The term $\cP(dx_t\mid\psi_t,\psi_{t+dt},x_t)$ describes the classical dynamics Eq.~\ref{dx}. The important point is that $dx_t$ does not only depend on $\psi_t$ but also on $\psi_{t+dt}$. The reason is that $\psi_{t+dt}$ depends on $d\xi_t$, and since $d\xi_t$ and $dW_t$ are correlated, it affects the probability of $dx_t$. When $d\xi_t$ and $dW_t$ are uncorrelated, the arrow from  $\psi_{t+dt}$ to $dx_t$ is absent and the graph depicts a classical hidden Markov model.
\begin{figure}
%\begin{center}
%\begin{tikzpicture}[>=Latex,node distance=13mm and 17mm, every node/.style={font=\small}]
%\foreach \k in {0,1,2,3}{
%  \node[draw,circle,minimum size=4mm,fill=soft] (z\k) at (2*\k,0) {$\psi_{\k}$};
%  \node[draw,rounded corners,minimum width=6mm,minimum height=7mm] (x\k) at (2*\k,-1.25) {$x_{\k}$};
%  \draw[->] (z\k) -- node[right] {$dW_{\k}$} (x\k);
%}
%
%\foreach \k/\j in {0/1,1/2,2/3}{
%  \draw[->,thick] (z\k) -- node[above] {$d\xi_{\k}$} (z\j);
%}
%
%% Box m_1 above psi_1
%\node[draw,rounded corners,minimum width=6mm,minimum height=7mm]
%  (m1) at (2,1.25) {$m_1$};
%
%% Upward arrow from psi_1 to m_1
%\draw[->] (z1) -- (m1);
% \draw[->] (z1) -- node[right] {$dW_1$} (m1);
%\end{tikzpicture}
%\end{center}
\begin{center}
\begin{tikzpicture}[
  >=Latex,
  every node/.style={font=\small},
  state/.style={draw,circle,minimum size=7mm,inner sep=0pt,fill=soft}
]

% psi nodes
\foreach \name/\x/\lbl in {
  z0/0/{t-dt},
  z1/2.5/{t},
  z2/5/{t+dt},
  z3/7.5/{t+2dt}
}{
  \node[state] (\name) at (\x,0) {$\psi_{\lbl}$};
}

% dx nodes and links
\foreach \left/\right/\name/\lbl in {
  z0/z1/x0/{t-dt},
  z1/z2/x1/{t},
  z2/z3/x2/{t+dt}
}{
  \draw[->,thick] (\left) -- (\right);

  \node[state] (\name) at ($(\left)!0.5!(\right)+(0,-1.05)$) {$dx_{\lbl}$};

  \draw[->,thick] (\left) -- (\name);
  \draw[->,thick] (\right) -- (\name);
}

% dm_t above, mirror symmetric to dx_t
\node[state] (m1) at ($(z1)!0.5!(z2)+(0,1.05)$) {$dm_t$};

\draw[->,thick] (z1) -- (m1);
\draw[->,thick] (z2) -- (m1);

% continuation dots left and right
\node at ($(z0)+(-1.2,0)$) {$\cdots$};
\node at ($(z0)+(-1.2,-1.05)$) {$\cdots$};

\node at ($(z3)+(1.2,0)$) {$\cdots$};
\node at ($(z3)+(1.2,-1.05)$) {$\cdots$};

\end{tikzpicture}
\end{center}\caption{The hybrid quantum classical dynamics depicted as a graphical model in discretized time with time discretization $dt$.  The dependence on $x$ is suppressed.}
\label{hmm}
\end{figure}

\subsubsection{Retrograde prediction of classical variables}
\label{section:retrograde}
Early work by \cite{tsang_time-symmetric_2009} developed a method to compute smoothing estimate for classical observables, ie. the retrograde prediction of a measurement outcome at intermediate time given all past and future observations $\cP(dm_t \mid x_{0:T})$.
The graphical structure in fig.~\ref{hmm} factorizes in a forward computation and a backward computation, so that 
\[
\cP(\psi_t,\psi_{t+dt},dm_t,dx_t\mid x_{0:T})\propto \underbrace{\cP(\psi_t \mid dx_{0:t-dt})}_\text{forward} \cP(dm_t,dx_t, \psi_{t+dt}\mid \psi_t) \underbrace{\cP(dx_{t+dt:T-dt} \mid \psi_{t+dt})}_\text{backward}
\]
\cite{tsang_time-symmetric_2009} and subsequent work \cite{tsang_optimal_2009,tsang_optimal_2010,gammelmark_past_2013} show that the forward and backward terms can be expressed as the 
 filtered density matrix $\rho_t^F$ and a (retrograde) effect operator $E_t$, collectively known as the {\em past quantum state}. The then are combined to compute the smoothed estimate $\cP(dm_t \mid x_{0:T})$. We derive the dynamical equation for the effect operator in Appendix~\ref{appendix:effect}. 
 
%Specifically, \cite{gammelmark_past_2013} show  that  the probability of outcome $m$, conditioned on observed outcomes $x_s$ at times $s<t$ and $s>t$, excluding $x_t$ is given by 
%\bea
%p_t(m) =\frac{\Tr\left(\hat \Omega_m \rho_t^F \hat\Omega_m^\dagger E_t\right)}{\sum_m \Tr\left(\hat \Omega_m \rho_t^F \hat\Omega_m^\dagger E_t\right)}\label{gammelmark}
%\eea
%We derive this expression in Appendix~\ref{appendix_gammelmark}. 

\subsubsection{Smoothed density matrix}
\label{section:quantumstatesmoothing}
 \cite{guevara_quantum_2015} consider a setting in which a smoothed density matrix can be defined. They introduce, in addition to the observed measurement record $x_{0:T}$, an unobserved measurement record \(u_{0:T}\). They assume  the saturated noise case, so that conditioning on both $x_{0:t}$ and $u_{0:t}$ the filtered quantum state is pure, $P_{\psi_{x_{0:t},u_{0:t}}}$. Because there is no remaining quantum uncertainty, conditioning on all observations $x_{0:T}$ cannot change the filtered state at time $t$ so that $\rho^S_t=\rho^F_t=P_{\psi_{x_{0:t},u_{0:t}}}$. Therefore, conditioned on $x_{0:T}$ alone, 
 the smoothed estimate is given by the filtered estimate, averaged over the smoothed posterior of $p(u_{0:t}|x_{0:T})$.
 \bea
\rho^S_t=\int du_{0:t} P_{\psi_{x_{0:t},u_{0:t}}} \cP(u_{0:t}\mid x_{0:T})\label{guevara}
\eea
The extra unobserved measurement record \(u_{0:T}\) and the saturated noise case have effectively shifted all uncertainty in the quantum state to uncertainty in $u_{0:t}$. Since this is a classical variable for which one can compute the smoothed estimate $\cP(u_{0:t}\mid x_{0:T})$ using the methods of \cite{tsang_time-symmetric_2009} and others discussed above.
%
%a simplification occurs because the smoothed state conditioned on $x_{0:T},u_{0:T}$ is equal to the filtered state $P_{\psi_{x_{0:T},u_{0:T}}}=P_{\psi_{x_{0:t},u_{0:t}}}$ and the smoothed posterior conditioned on $x_{0:T}$ is given by
%\footnote{Consider observations $x_{0:T}$ that fully determine the pure filtered state $\rho_t^F=P_{\psi_{x_{0:t}}}$. Then the smoothed state
%\beaa
%\rho_S(t) &=&\frac{1}{p(x_{0:T}) } \int d\psi_{t} du_{0:T}P_{\psi_t} \cP(\psi_t, x_{0:T},u_{0:T}) \\
%&=&\frac{1}{p(x_{0:T}) }\int d\psi_{t} du_{0:t}P_{\psi_t} \cP(\psi_{t},x_{0:t},u_{0:t})\cP(x_{t:T}\mid \psi_t,x_t,u_t)\\
%&=&\frac{1}{p(x_{0:T}) }\int du_{0:t} P_{\psi_{x_{0:t},u_{0:t}}} p(x_{0:t},u_{0:t}) p(x_{t:T}\mid x_t,u_t) =\int du_{0:t} P_{\psi_{x_{0:t},u_{0:t}}} p(u_{0:t}\mid x_{0:T})
%\eeaa
%}
$\cP(u_{0:t}\mid x_{0:T})$ can be estimated by standard particle filtering\footnote{
Since the latent state conditioned on $x_{0:t},u_{0:t}$ is pure, the joint probability of the classical variables is Markovian: $p(x_{0:T},u_{0:T}) = p(x_{0:t},u_{0:t}) p(x_{t:T},u_{t:T}\mid x_t,u_t)$.}.
%\footnote{
%Consider the probability of the measurement trajectory
%\beaa
%p(x_{0:T})&=&\int d\psi_{0:T} \cP(x_{0:T},\psi_{0:T}) = \int d\psi_{0:T} \cP(x_{0:t},\psi_{0:t})\cP(x_{t+dt:T},\psi_{t+dt:T}|x_t,\psi_t)
%=\int d\psi_t \cP(x_{0:t},\psi_t) \cP(x_{t+dt:T}|x_t,\psi_t)
%\eeaa
%In the saturated noise case with pure initial state, the trajectory $x_{0:t}$ determines uniquely the quantum state $\psi_{x_{0:t}}$, so that the measurement process $x$ becomes Markovian:
%\beaa
% p(x_{0:T})&=&p(x_{0:t})p(x_{t+dt:T}|x_t)
%\eeaa
%with $p(x_{0:t})=\cP(x_{0:t},\psi_{x_{0:t}})$ and $p(x_{t+dt:T}|x_t)= \cP(x_{t+dt:T}|x_t,\psi_{x_{0:t}})$.
%\beaa
%\rho_S(t) &=& \int d\psi_{0:T} \psi_t \psi_t^\dagger \cP(\psi_{0:T}|x_{0:T}) =\frac{1}{p(x_{0:T}) } \int d\psi_{0:T} \psi_t \psi_t^\dagger \cP(\psi_{0:T}, x_{0:T}) \\
%&=&\frac{1}{p(x_{0:T}) }\int d\psi_{0:T} \psi_t \psi_t^\dagger \cP(\psi_{0:t},x_{0:t})\cP(\psi_{t+dt:T},x_{t+dt:T}|\psi_t,x_t)= \frac{1}{p(x_{0:T}) }\int d\psi_t\psi_t \psi_t^\dagger \cP(\psi_{t},x_{0:t})\cP(x_{t+dt:T}|\psi_t,x_t)\\
%&=&\frac{1}{p(x_{0:T}) }\psi_{x_{0:t}}\psi_{x_{0:t}}^\dagger p(x_{0:t}) p(x_{t+dt:T}|x_t) =\psi_{x_{0:t}}\psi_{x_{0:t}}^\dagger
%\eeaa
%}
The approach was  subsequently developed in
Refs.~\cite{laverick_quantum_2019,guevara_completely_2020,chantasri2021unifying, laverick_quantum_2023}.

%Quantum smoothing is considered quite different from classical smoothing in the literature, and two reasons are often stated. One is that the quantum back-action produced by measurement; the other reason that is often heard is the non-commutation of operators. However, we argue that these are not the essential issues. Classically, there is also back-action simply due to Bayes' rule that changes a prior to a posterior as a result of an observation. 
%The difference is also not due to non-commuting operators, because, as we argue in this paper, the underlying structure can be viewed as a classical probability distribution and the use of quantum operators is not necessary. 

Our Bayesian formulation allows the definition of the {\em smoothed quantum state} Eq.~\ref{rhoS1} for arbitrary hybrid quantum-classical systems and thus generalizes the definition Eq.~\ref{guevara} of \cite{guevara_quantum_2015}, and generalizes it to the non saturated noise case and does not require the introduction of a latent unobserved measurement record. Eq.~\ref{rhoS1} has been proposed previously for the case where a Gaussian description of the quantum state is appropriate in the context of magnetic-field sensing \cite{zhang_estimating_2020}.

Using Eq.~\ref{bE} we can express $\rho_t^S$ in terms of the filtered distribution and the effect operator as
\bea
\rho^S_t= \int d\phi_t P_{\phi_t} \cP^{(1)}(\phi_t\mid x_{0:T} )=\frac{\int d\phi_t P_{\phi_t}\Tr(P_{\phi_t} E_t)\cP^{(1)}(\phi_t|x_{0:t})}{\ell_t}\label{smoothed}
\eea

We summarize this section as follows. Smoothing has two aspects: retrodiction of measurement outcomes $\cP(dm_t\mid x_{0:T})$ and the smoothed distribution of the hidden state $\cP(\psi_t\mid x_{0:T})$. 
For classical hidden Markov models, the process noise $d\xi$ and the observation noise $dW$ are often assumed independent. In this case $\cP(m_t\mid \psi_t,\psi_{t+dt})=\cP(m_t\mid \psi_t)$ and 
\begin{align}
\cP(m_t\mid x_{0:T})
=
\int d\psi_t\,
\cP(m\mid \psi_t)
\cP(\psi_t\mid x_{0:T})=
\operatorname{Tr}(M_m\rho_t^{S})
\end{align}
where in the last step we assumed 
that measurement is described by a POVM element $M_m$ so that
$
\cP(m\mid \psi_t)
=
\operatorname{Tr}(M_m P_{\psi_t})
$.  In other words, the smoothed density matrix (or more general, the smoothed distribution $\cP(\psi_t\mid x_{0:T})$) is {\em sufficient} to predict the measurement outcome. The structure is
\begin{equation}
x_{0:T}
\longrightarrow
\rho_t^{S}
\longrightarrow
\cP(m_t\mid x_{0:T}).\label{screening}
\end{equation}

%Papers citing Budini 2017 Smoothed quantum-classical states in time-irreversible hybrid dynamics:
%QuantumÐclassical hybrid dynamics: Coupling mechanisms and diffusive approximation
%AA Budini - Open Systems & Information Dynamics, 2024 - World Scientific NO SMOOTHING
%Entropic relations for retrodicted quantum measurements
%AA Budini - Physical Review A, 2018 - APS INGEWIKKELD

%A. Chantasri, I. Guevara, H. M. Wiseman, Quantum state smoothing: why the types of observed and unobserved
%measurements matter, New Journal of Physics 21 (8) (2019) 083039. doi:10.1088/1367-2630/ab396e.
%[35] I. Guevara, H. M. Wiseman, Completely positive quantum trajectories with applications to quantum state smoothing,
%Physical Review A 102 (5) (2020) 052217. doi:10.1103/PhysRevA.102.052217.

Fig.~\ref{hmm} clearly shows in the quantum case when $d\xi$ and $dW$ are correlated, retrodiction of classical variables such as $dm_t$ requires two different objects with information about $\psi_t$ and $\psi_{t+dt}$. This information is provided by $\rho_t^F$ and $E_t$, respectively, ie. the past quantum state. At the same time, it is perfectly sound to define a smoothed density matrix Eq.~\ref{rhoS1} based on the smoothed distribution $\cP(\psi_t\mid x_{0:T})$. It should provide a more accurate estimate of the quantum state than the filtered density matrix, because it is based on more data, but from the above discussion it is clear that it cannot be used to predict measurement outcomes. In fact \cite{gammelmark_past_2013} show that using $\rho_t^S$ to predict measurement outcomes would give an incorrect result, which is not surprising given the above analysis.
We therefore conclude that  {\em the essential difference between quantum smoothing and classical smoothing is that the quantum noise $d\xi$ that defines the prior quantum proces and the observation noise $dW$ are correlated.}  This is typically not the case in classical hidden Markov models.

\section{Particle filtering and smoothing}
\label{section:filtering}

Since the hybrid dynamics is represented by an ordinary stochastic
process, the filtered distribution of the latent wave function can be
approximated by a weighted ensemble of stochastic trajectories. We use a
time discretization $dt$ and define $
t_k = k dt$ 
and denote the observed classical increments by $
d x_k = x_{k+1}-x_k$ and define $d\psi_k =\psi_{k+1}-\psi_k$. 

We write the unravelings  Eqs.~\eqref{sse3} and~\eqref{dx} in the form
\beaa
d \psi_k
&=&
F(\psi_k,x_k,t_k)d t
+
G_a(\psi_k,x_k,t_k)\,d\xi_{a,k},
\label{eq:pf_general_psi}
\\
dx_{i,k}
&=&
h_i(\psi_k,x_k,t_k)\,dt
+
dW_{i,k}
\eeaa
where
\beaa
F(\psi,x,t)
&=&
\left[
-iH
-\frac{1}{2}Q_{ab}
\left(
L_b^\dagger L_a
-2c_b^*L_a
+c_ac_b^*
\right)
\right]\psi,
\\
G_a(\psi,x,t)
&=&
(L_a-c_a)\psi,
\\
h_i(\psi,x,t)
&=&
f_i^c
+
\Gamma_{ia}c_a
+
\Gamma_{ia}^*c_a^*.
\eeaa
and $dW_{i,k}, d\xi_{a,k}$ real and satisfy Eqs.~\eqref{fp2lindblad1b}.

The filtered estimate $\cP_k(\psi\mid x_{0:k})$ satisfies the standard filtering update equation
\bea
\cP_{k+1}(\psi'\mid x_{0:k+1}) &=& \cP_{k+1}(\psi'\mid x_{0:k},dx_k)=\int d\psi \cP_k(\psi\mid x_{0:k}) \cP(d\psi\mid dx_k, \psi,x_k)\label{filtering1}
\eea
with $\psi'=\psi+d\psi$. 
The particle filter approximates $\cP_k(\psi \mid x_{0:k})$ by an estimate $\hat\cP_k(\psi \mid x_{0:k})$ using $N$ particles $\psi_k^{(r)}$ with associated weights $w_k^{(r)}$ such that $\sum_r w_k^{(r)}=1$:
\beaa
\hat\cP_k(\psi \mid x_{0:k})=\sum_{r=1}^N w_k^{(r)}\delta\left(\psi-\psi_k^{(r)}\right)
\eeaa
The filtering update Eq.~\ref{filtering1} becomes
\beaa
\hat\cP_{k+1}(\psi' \mid x_{0:k+1})
= \sum_r w_r^{(r)} \cP(d\psi\mid dx_k, \psi_k^{(r)},x_k)
\eeaa
The filtering recursion consists of two components: 1) an update of each particle $\psi_{k+1}^{(r)}=\psi_{k}^{(r)}+d\psi_k^{(r)}$ by sampling an increment  $d\psi_{k}^{(r)}$ from $\cP(d\psi\mid dx_k,\psi_k^{(r)},x_k)$ and 2) an update of its weight with the likelihood of the observed increment:  $w^{(r)}_{k+1} \propto w^{(r)}_k \cP(dx_{k}\mid \psi_k^{(r)},x_k)$. We discuss these steps in detail.

\paragraph{Updating the particle}
For small $dt$, $\cP(d\psi_{k},dx_{k}\mid \psi_k,x_k)$ 
%and $p(d\psi_{k}|dx_{k},\psi_k,x_k)$ are 
is Gaussian distributed which allows us to sample $d\psi_k$ conditioned on $dx_k$. Equivalently, we write $d\psi_k = F_kdt +G_{a,k} d\xi_{a,k}$ with  $d\xi_{a,k}, dx_k$ jointly Gaussian with  mean $(0, h_{k} dt)$ and covariance matrix $D=\left(\begin{matrix} Q & \Gamma \\ \Gamma^\dagger & C\end{matrix}\right)$. Conditioned on $dx_k$, $d\xi_k$ is Gaussian distributed as (see footnote 
on page \pageref{gaussian_conditional})
\bea
d\xi_{k} \sim \cN(d\xi_k| \mu_k, S) \qquad \mu_k = \Gamma^\dagger C^{-1} (dx_k -h(\psi_k,x_k,t_k)dt)\qquad S = Q-\Gamma^\dagger C^{-1} \Gamma\label{dxi1}
\eea
Positivity of $D$ implies
$
S\geq 0
$ because of Eq.~\ref{noise_bound}.
$d\xi_k^{(r)}$  can be sampled according to
\begin{equation}
d\xi_k^{(r)}
=
\Gamma^\dagger C^{-1}dI^{(r)}_k
+S^{1/2}\sqrt{\Delta t}\,\epsilon_k^{(r)},\qquad\epsilon_k^{(r)}\sim\mathcal{N}(0,I)\qquad dI_k^{(r)}=dx_k -h(\psi_k^{(r)},x_k,t_k)dt
\label{eq:pf_quantum_noise_sample}
\end{equation}
and the propagation equation for particle \(r\) is
\begin{equation}
\psi_{k+1}^{(r)}
=
\psi_k^{(r)}
+
F_k^{(r)}d t
+
\sum_a G_{a,k}^{(r)}
d\xi_{a,k}^{(r)},
\label{eq:pf_particle_propagation}
\end{equation}
where $
F_k^{(r)}
=
F(\psi_k^{(r)},x_k,t_k),
G_{a,k}^{(r)}
=
G_a(\psi_k^{(r)},x_k,t_k)$.

\paragraph{Updating the weights} 
Using the conditional quantum transition density as the proposal
distribution, the unnormalized particle weights satisfy
\beaa
\widetilde{w}_{k+1}^{(r)}
=
w_k^{(r)}
\cP\left(
d x_k
\mid
\psi_k^{(r)},x_k
\right).
\eeaa
Since $\cP(dx_k | \psi_k,x_k)$ is Gaussian with mean $h_k dt$ and covariance $Cdt$ we obtain
\begin{align}
\log\widetilde{w}_{k+1}^{(r)}
=
\log w_k^{(r)}
-\frac{1}{2d t}
\left(d I_k^{(r)}\right)^T
C^{-1}
dI_k^{(r)}.
\label{eq:pf_log_weight}
\end{align}
where we ignore the $\log \det C$ term since it is independent of the sample and drops out after normalization and $d I_k^{(r)}$ is given by Eq.~\ref{eq:pf_quantum_noise_sample}.
The normalized weights are
\begin{equation}
w_{k+1}^{(r)}
=
\frac{
\widetilde{w}_{k+1}^{(r)}
}{
\sum_{s=1}^{N}
\widetilde{w}_{k+1}^{(s)}
}.
\label{eq:pf_normalized_weights}
\end{equation}
The weights are initialized as $w_{k=0}^{(r)}=1/N$. 
\paragraph{Filtered density matrix}
The particle filtering algorithm can be applied to either the normalized or unnormalized unravelings. For the normalized unraveling, the filtered estimate of the density matrix is
\begin{equation}
\hat\rho_{k}
=
\sum_{r=1}^{N}
w_k^{(r)}
P_k^{(r)},
\qquad
P_k^{(r)}
=
\phi_k^{(r)}
\phi_k^{(r)\dagger},\qquad \operatorname{Tr}P_k^{(r)}=1
\label{eq:pf_filtered_density_nonlinear}
\end{equation}
Alternatively, because of Eq.~\ref{eq:rho-sigma-section}, we can also estimate $\rho_k$ by first estimating $\sigma_t$ for the unnormalized linear unraveling
\begin{equation}
\hat\rho_k = \frac{\hat\sigma_k}{\Tr \hat\sigma_k}\qquad \hat\sigma_{k}
=
\sum_{r=1}^{N}
v_k^{(r)}
P_k^{(r)},
\qquad
P_k^{(r)}
=
\psi_k^{(r)}
\psi_k^{(r)\dagger},
\label{eq:pf_filtered_density_linear}
\end{equation}
with $v_k^{(r)}$ the filtering weights and $\Tr P_k^{(r)}\ne 1$. These estimates are different because their samples come from different distributions and have different weights. 

\paragraph{Efficiency and resampling}
Since the particle weights are updated by multiplication, one should expect that some particles will exponentially dominate over all other particles for large $k$. 
The efficiency of the particle filtering method is monitored by means of the effective sample
size
\begin{equation}
N_{\mathrm{eff}}
=\left(
\sum_{r=1}^{N}
\left(w_k^{(r)}\right)^2\right)^{-1}
\label{eq:pf_effective_sample_size}
\end{equation}
When \(N_{\mathrm{eff}}\) falls below a prescribed threshold, the
particles can be resampled according to their normalized weights and the
new weights are set to $
w_k^{(r)}=1/N$. We will not consider resampling in this paper. 

%The complete particle-filtering recursion is therefore as follows:
%\begin{enumerate}
%    \item For every particle, compute
%    \[
%    c_{a,k}^{(r)},
%    \qquad
%    h_k^{(r)},
%    \qquad
%    dI_k^{(r)}
%    =
%    d x_k-h_k^{(r)}d t.
%    \]
%
%    \item Update the unnormalized importance weight according to
%    \[
%    \widetilde{w}_{k+1}^{(r)}
%    =
%    w_k^{(r)}
%    \mathcal{N}\left(
%    d x_k;
%    h_k^{(r)}d t,
%    C d t
%    \right).
%    \]
%
%    \item Draw the residual quantum noise and set
%    \[
%    d\xi_k^{(r)}
%    =
%    \Gamma^\dagger C^{-1}dI_k^{(r)}
%    +
%    S^{1/2}\sqrt{d t}\,
%    \epsilon_k^{(r)}.
%    \]
%
%    \item Propagate \(\psi_k^{(r)}\) using
%    Eq.~\eqref{eq:pf_particle_propagation}.
%
%    \item Normalize the importance weights and resample when
%    \(N_{\mathrm{eff}}\) is sufficiently small.
%\end{enumerate}

\paragraph{Saturated noise case}
In the saturated-noise case 
$S=
Q-
\Gamma^\dagger C^{-1}\Gamma=0
$
(see section~\ref{section:noise_covariance}).
Eq.~\ref{dxi1} then shows that 
the quantum-noise increment is  completely determined by the
observed increment:
\beaa
d\xi_k^{(r)}
=
\Gamma^\dagger C^{-1}dI_k^{(r)}
\eeaa
There is no residual quantum noise to sample. Consequently, for a known
initial pure state, the observed trajectory determines the quantum
trajectory recursively. In this limit a single particle is sufficient,
apart from uncertainty in the initial state or in unknown model
parameters.

\paragraph{Particle smoothing}
For particle smoothing, the complete particle genealogy is retained
during the forward filtering pass. In the simplest case one runs the particle filter until the final time $T$ without resampling. Then the smoothed estimates at any intermediate time are given by 
\begin{equation}
\hat\rho^S_k
=
\sum_{r=1}^{N}
w_{T}^{(r)}
P_k^{(r)}\label{eq:pf_smoothed_density}
\end{equation}
and similar for $\hat\sigma^S_k$. 
When using resampling, more advanced methods can be employed, such as  the Forward Filter Backward Simulation (FFBSi) method \cite{godsill2004monte,douc2011sequential}. 
See also this tutorial \cite{doucet2009tutorial}.

\section{Numerical examples}
\label{numerical}
In this section, we present two numerical examples that illustrate 
distinctive features of the Bayesian formulation. 
\bi
\item 
The Bayesian formulation performs smoothing directly on the probability distribution over latent quantum states and trajectories. This differs from the past-quantum-state formalism, in which forward and backward operators are combined to obtain retrodictive probabilities for specified measurements but do not, in general, define a universal smoothed density matrix. The first example demonstrates that, for a hidden classical process coupled to a quantum system, conditioning on the complete observation record improves both reconstruction of the classical process and the associated smoothed quantum state compared with causal filtering.
\item 
The posterior distribution over latent quantum states contains inferential information that is not retained by its density-matrix second moment. This does not contradict the operational completeness of the density matrix for predicting measurements on the quantum system: the additional information concerns the latent trajectory in the specified unraveling. The second example demonstrates this explicitly: the posterior over latent pure-state trajectories can be strongly multimodal even when its density-matrix average is close to maximally mixed.
The distribution over quantum trajectories was previously studied by \cite{weber_mapping_2014} for transmon qubits. 
%Uses method called quantum Bayes rule which is interleaving unitary evolution with measurements from Korotkov 2011,1999 \url{http://arxiv.org/abs/1111.4016} does not use KS. Observes that the ensemble prediction is mixed, while individual trajectories stay pure.
%Constructs most probable quantum trajectories between initial and final state for transmon qubit. 

\ei

% Suggested subsection to insert after Section 6.1.
% Figure filenames below are placeholders; change them to the names used when
% exporting the MATLAB figures.

\subsection{Filtering and smoothing a hidden quantum-classical process}
\label{sec:telegraph-example}
We illustrate an inference problem in which a coupled quantum and classical process $\psi_t,s_t$ is 
 hidden and is observed only indirectly through continuous measurements $x_t$. 
This provides a simple
example in which particle smoothing has a clear advantage over
causal filtering. 

We consider a classical random telegraph process
$s_t\in\{-1,+1\}$
with symmetric switching rate $\lambda$ so that $\Pr(s_{t+dt}=-s_t\mid s_t)=\lambda\,dt+o(dt)$. 
The telegraph variable modulates the Hamiltonian of a single qubit according to
\begin{equation}
    H_t
    =
    \frac{\omega(s_t)}{2}\sigma_y,
    \qquad
    \omega(s_t)=\omega_0+\Delta\omega\,s_t .
    \label{eq:telegraph-H}
\end{equation}
Thus the two classical states correspond to two different Rabi frequencies
$\omega_\pm=\omega_0\pm\Delta\omega$.

The qubit is continuously monitored in the $\sigma_z$ basis with total
measurement strength $\kappa$ so that $L=\sqrt{\kappa} \sigma_z$ and detector efficiency $\eta_{\rm eff}$.  Eqs.~\ref{sse3} and~\ref{dx} become
\beaa
d\psi_t
&=&-iH_t\psi_t\,dt
-\frac{\kappa}{2}\left(\sigma_z-\langle\sigma_z\rangle_t\right)^2\psi_t\,dt+
\sqrt{\kappa}\left(\sigma_z-\langle\sigma_z\rangle_t\right)\psi_t\,d\xi_t
\label{eq:telegraph-sse}\\
dx_t&=&2\Gamma \sqrt{\kappa}\,\langle\sigma_z\rangle_t\,dt+dW_t \label{eq:telegraph-record}
\eeaa
with $Q=C=1$ and  $\Gamma=\sqrt{\eta_\text{eff}}$.

For $\eta_{\rm eff}=1$  the quantum measurement backaction is completely determined by the
observed record, but the classical telegraph trajectory $s_{0:T}$ remains
hidden.  For $\eta_{\rm eff}<1$, both the telegraph trajectory and the
unobserved measurement backaction must be inferred.
The Bayesian state of the inference problem is therefore the joint posterior
over the hybrid trajectory. For filtering and smoothing the posterior is conditioned on the past measurements or all measurements, respectively:
\begin{equation}
    \cP(\psi_{0:t},s_{0:t}\mid x_{0:t})
\qquad \text{or}\qquad    \cP(\psi_{0:T},s_{0:T}\mid x_{0:T})
\end{equation}
  We approximate these distributions by particles $(\psi_t^{(r)},s_t^{(r)})$ with corresponding weights $w_t^{(r)}$ as described in section~\ref{section:filtering}. Each
particle carries a candidate telegraph history and the corresponding quantum
trajectory. 
The filtered probability that the telegraph variable is in the $+1$ state is
then
\begin{equation}
    p_t^F
    :=
    \cP(s_t=+1\mid Y_{0:t})
    \simeq
    \sum_{r=1}^N
    w_t^{(r)}
    {\bf 1}\!\left[s_t^{(r)}=+1\right].
    \label{eq:telegraph-filter}
\end{equation}
For the simple no-resampling path-space smoother used here, the same stored
particle histories are reweighted by their final weights,
\begin{equation}
    p_t^S
    :=
    \cP(s_t=+1\mid Y_{0:T})
    \simeq
    \sum_{r=1}^N
    w_T^{(r)}
    {\bf 1}\!\left[s_t^{(r)}=+1\right].
    \label{eq:telegraph-smoother}
\end{equation}
The smoothed posterior probability that a transition occurred between
$t_k$ and $t_{k+1}$ is obtained similarly,
\begin{equation}
    q_k^S
    =
    \sum_{r=1}^N
    w_T^{(r)}
    {\bf 1}\!\left[s_{k+1}^{(r)}\neq s_k^{(r)}\right].
    \label{eq:telegraph-switchprob}
\end{equation}

We estimated the difference in quality between filtering and smoothing. We generate instances, where the classical telegraph signal is always the same with $s(0)=1, s(T)=-1$ and switching at $t=T/2$, and the quantum trajectory is initialized uniform at random from the Bloch sphere. Note, that in this setup the dynamics switches from high to low frequency, which is harder to detect than from low to high frequency. 
From this, the classical measurement sequence $x_{0:T}$ is generated using Eqs.~\ref{sse3} and~\ref{dx} using normalized unravelings $\eta=1$ and a specified value of the detector efficiency $\eta_\text{eff}$. From the measurement sequence, the telegraph signal and the quantum state is reconstructed using this same value of $\eta_\text{eff}$. So we do not consider detector-efficiency mismatch. 
Figure~\ref{fig:telegraph-filter-smoother} shows some typical results for $\eta_\text{eff}=0.4,0.8, 1$.
\begin{figure}
\bc
\includegraphics[width=0.3\textwidth]{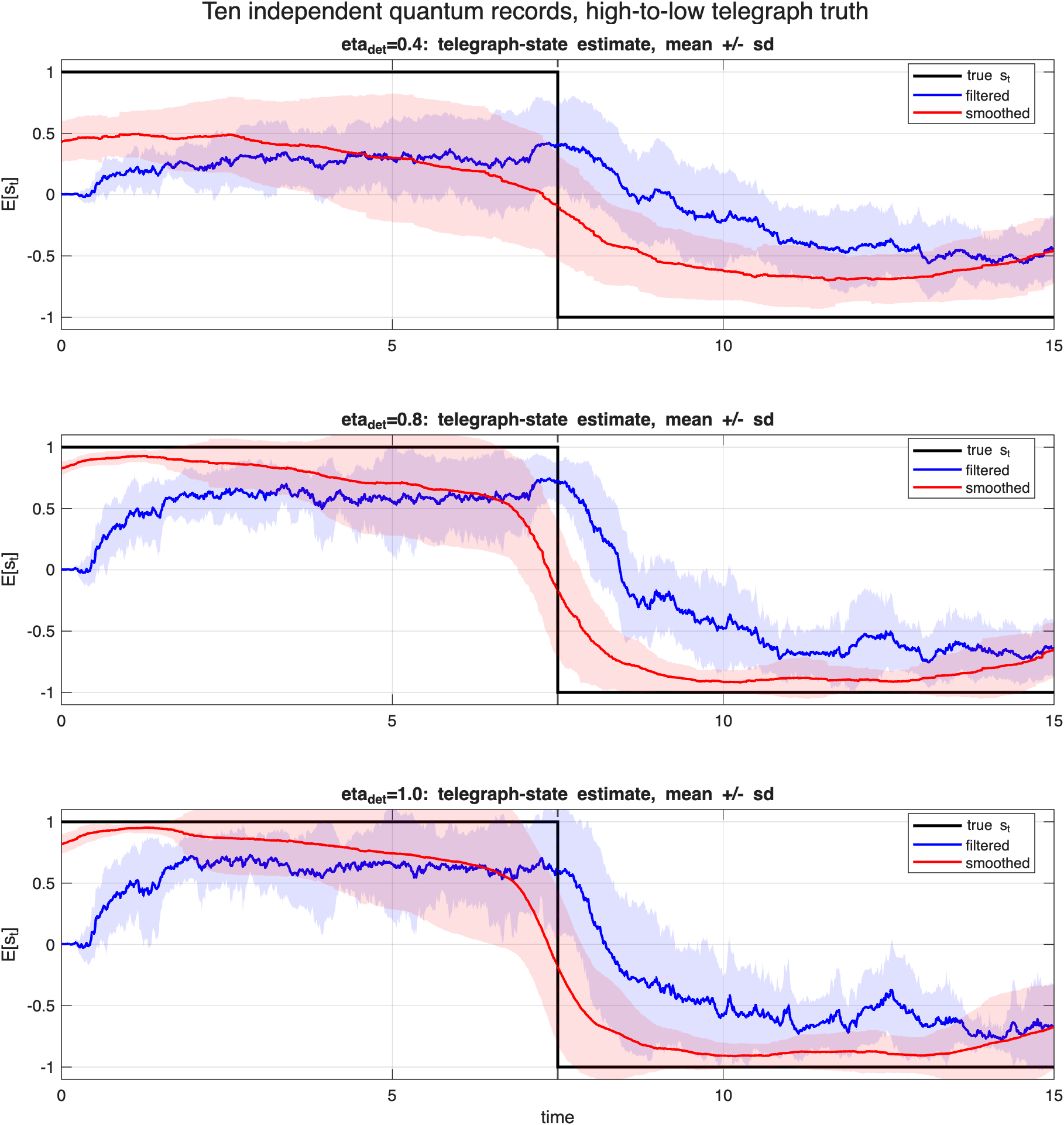}
\includegraphics[width=0.3\textwidth]{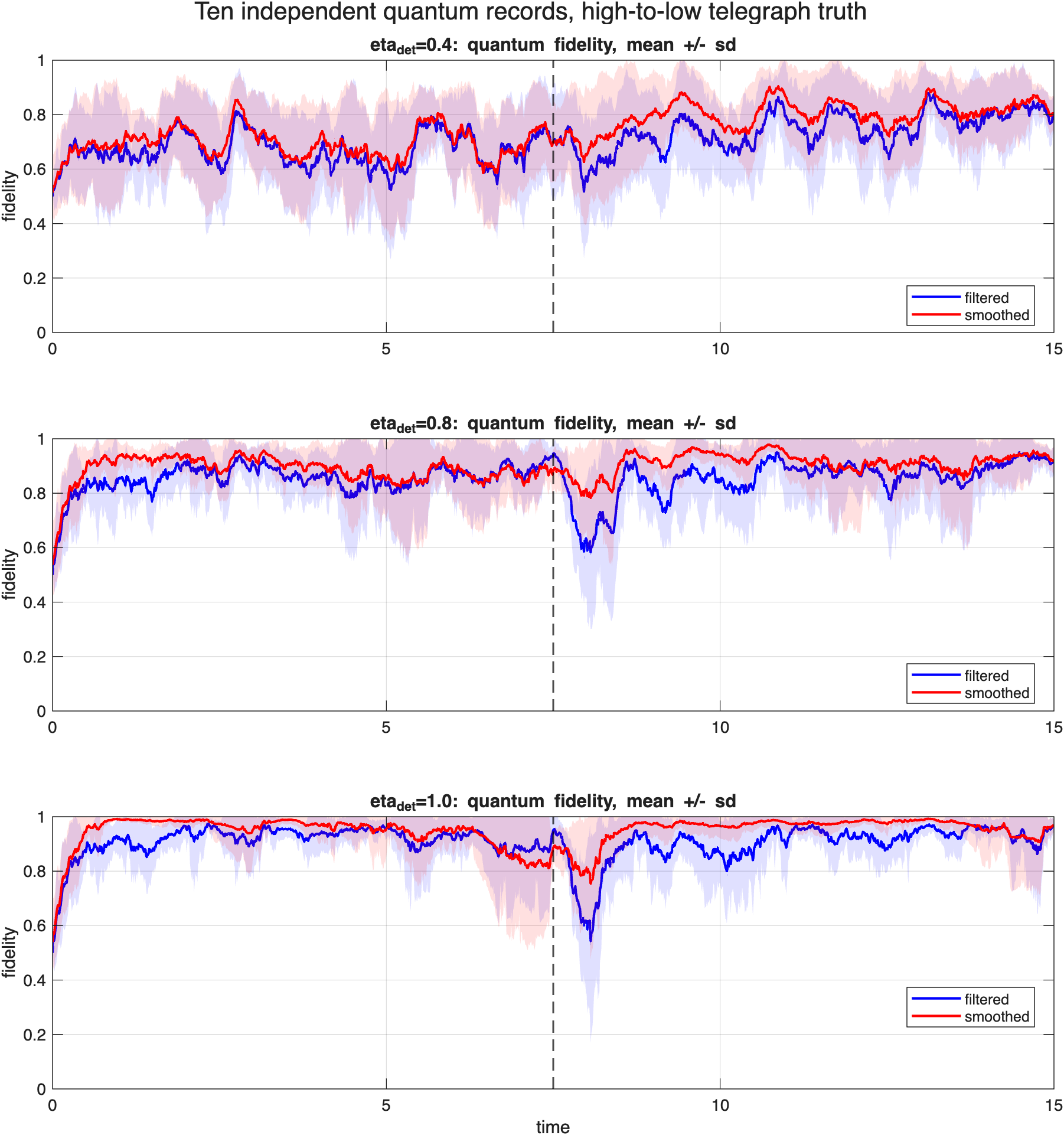}
\ec
\caption{Filtering and smoothing of a hidden classical telegraph process and quantum state
through a continuously monitored qubit for $\eta_\text{eff}=0.4,0.8,1.0$. Left: expected telegraph signal and ground truth.
Right: quantum fidelity.
Parameters are $
    dt=0.01, 
    T=15,
    \lambda=0.12,
    \omega_0=4,
    \Delta\omega=2,
    \kappa=1$ and 
$N=\num{2e4}$ particles. Average over 10 runs with identical telegraph signal and randomized initial quantum state. 
}
\label{fig:telegraph-filter-smoother}
\end{figure}
% telegraph_field_particle_smoothing_qubit_fidelity_4panel.m
The smoother reconstructs the telegraph signal and the quantum state much more accurately  as summarized in Table~\ref{tab:filter_smooth}.
Obviously, both the filtering and smoothing accuracies increase with
increasing $\eta_\text{eff}$ and the improvement of smoothing over filtering is significant for all $\eta_\text{eff}$. 
The filter detects the switch time too late, and only after sufficient
future measurements have accumulated, whereas the
smoother uses those measurements retrospectively to localize the transition
close to the true switching time.
\begin{table}[ht]
\centering
\begin{tabular}{c|cc|cc}
\hline
$\eta_{\rm det}$ 
& \multicolumn{2}{c|}{$P(\mathrm{correct})$}
& \multicolumn{2}{c}{Fidelity} \\
& Filtered & Smoothed & Filtered & Smoothed \\
\hline
0.4 & $0.6281 \pm 0.0601$ & $0.7273 \pm 0.0943$
    & $0.7060 \pm 0.0322$ & $0.7462 \pm 0.0356$ \\
0.8 & $0.7399 \pm 0.0450$ & $0.8875 \pm 0.0674$
    & $0.8536 \pm 0.0378$ & $0.8997 \pm 0.0402$ \\
1.0 & $0.7550 \pm 0.0474$ & $0.8931 \pm 0.0714$
    & $0.9059 \pm 0.0279$ & $0.9483 \pm 0.0338$ \\
\hline
\end{tabular}
\caption{Comparison of time-averaged filtered and smoothed estimates as a function of detection efficiency $\eta_{\rm eff}$. Parameter setting as in fig.~\ref{fig:telegraph-filter-smoother}.}
\label{tab:filter_smooth}
\end{table}

%\begin{figure}
%\bc
%\includegraphics[width=0.6\textwidth]{telegraph_filter_smoothing}
%\ec
%\caption{Filtering and smoothing of a hidden classical telegraph process and quantum state
%through a continuously monitored qubit. Subplots are numbered 1-4. 1: true telegraph state $s_t$
%(black), filtered posterior mean $\mathbb{E}[s_t\mid Y_{0:t}]$ (blue), and
%smoothed posterior mean $\mathbb{E}[s_t\mid Y_{0:T}]$ (red). 
%The filter detects the switches only after
%future measurements have accumulated sufficient evidence, whereas the
%smoother uses those measurements retrospectively to localize the transitions
%close to their actual times.
%2:
%  filtered, and smoothed fidelities $\av{\psi_t | \rho_t^{F,S}\psi_t}$. 
%3: smoothed posterior switch density $q_k^S/dt$; dashed vertical lines mark
%the true transitions and the horizontal dotted line is the prior switching
%rate $\lambda$. 4: effective sample size divided by the number of
%particles for the no-resampling path-space importance sampler.
%Parameters are $
%    dt=0.01, 
%    T=15,
%    \lambda=0.12,
%    \omega_0=4,
%    \Delta\omega=2,
%    \kappa=1,
%    \eta_{\rm eff}=0.9$ and 
%$N=\num{2e4}$ particles.  }
%\label{fig:telegraph-filter-smoother}
%\end{figure}
%% telegraph_field_particle_smoothing_qubit_fidelity_4panel.m

This numerical example is closely related to \cite{madsen_quantum_2021}. Their example reconstructs a
fluctuating magnetic field discretized in $n$ values, and sensed by a single quantum spin. 
%They describe a time-dependent
%perturbation by a classical hidden Markov model, with classical states $n$, couple it to a continuously
%monitored quantum system, and combine classical forward--backward inference
%with the past-quantum-state formalism with density matrix $\rho_n(t)$ and effect matrix $E_n(t)$. 
They then compute the smoothed probability 
estimate of 
the classical HMM state as $P_S(n,t) \propto \Tr(\rho_n(t) E_n(t))$.  
%For a finite
%classical state space, such as the two-state telegraph process used in
%section~\ref{sec:telegraph-example}, an explicit hybrid hidden-Markov
%recursion is therefore a natural alternative to particle filtering and can exploit the finite-state structure very efficiently because it exploits the special finite-state
%structure. 
Instead, our particle filtering approach  reconstructs the smoothed distribution over both the quantum trajectories and the classical input $\cP(\psi_{0:T},s_{0:T}\mid x_{0:T})$.

\subsection{Bimodal trajectory posteriors under partial environmental monitoring}
\label{sec:bimodal-example}

In the second example, we consider a qubit coupled to two continuously monitored environmental
channels. The first channel is accessible to the observer and produces the
measurement record $x_{0:T}$, whereas the second channel produces an
unobserved record $u_{0:T}$. This is the partially observed setting considered
in quantum state smoothing by \cite{guevara_quantum_2015} and described in section~\ref{section:quantumstatesmoothing}: if both records were known, the conditional state
would remain pure and would be described by a stochastic quantum trajectory,
while an observer with access only to $x_{0:T}$ must average over the
unobserved record.

The purpose of this experiment is to illustrate the distinction between the
posterior distribution over the latent pure-state trajectories, $
\cP\left(\psi_{0:T},u_{0:T}\mid x_{0:T}\right),
$
and the smoothed density matrix Eq.~\ref{smoothed}. 
The example is designed so that the posterior over pure states develops two
well-separated modes, although their density-matrix average is close to a
maximally mixed state. 

The qubit is initialized in the state

\begin{equation}
\ket{\psi_0}
= \ket{+x}=
\frac{\ket{+z}+\ket{-z}}{\sqrt{2}}.
\label{eq:bimodal-initial-state}
\end{equation}
The initial state is assumed to be known. 
We consider two diffusive measurement channels. The observed channel
continuously monitors $\sigma_x$ with measurement strength $\kappa_x$,
whereas the hidden channel continuously monitors $\sigma_z$ with strength
$\kappa_z$. Conditioned on both measurement records, the normalized pure
state satisfies
\begin{align}
d\psi_t
=&
-\frac{\kappa_x}{2}
\left(
\sigma_x-\langle\sigma_x\rangle_t
\right)^2
\psi_t dt
+
\sqrt{\kappa_x}
\left(
\sigma_x-\langle\sigma_x\rangle_t
\right)
\psi_t d\xi_t^{(x)}
\nonumber\\
-&
\frac{\kappa_z}{2}
\left(
\sigma_z-\langle\sigma_z\rangle_t
\right)^2
\psi_t dt
+
\sqrt{\kappa_z}
\left(
\sigma_z-\langle\sigma_z\rangle_t
\right)
\psi_t d\xi_t^{(u)}
\label{eq:bimodal-two-channel-sse}
\end{align}
where
$
\langle\sigma_j\rangle_t=
\langle\psi_t|\sigma_j|\psi_t\rangle,
j\in{x,z}$.
The Wiener increments  $d\xi_t^{(x,u)}$ are real and independent: $\av{d\xi_t^{(i)}d\xi_t^{(j)}}=\delta_{ij}dt, i,j = x,u$. 
The measurement increments are
\beaa
dx_t&=&2\sqrt{\kappa_x}
\langle\sigma_x\rangle_tdt+dW_t^{(x)},\label{eq:bimodal-observed-record}\\
du_t&=&2\sqrt{\kappa_z}
\langle\sigma_z\rangle_tdt
+
dW_t^{(u)}.
\label{eq:bimodal-hidden-record}
\eeaa
Because of the saturated noise setting: $dW_t^{(i)}=d\xi_t^{(i)}, i=x,u$. Therefore we can eliminate $dW_t^{(i)}=d\xi_t^{(i)}$ in 
Eq.~\ref{eq:bimodal-two-channel-sse} in favor of $dx_t, du_t$ and 
the quantum state trajectory is therefore determined by the pair
$(x_{0:T},u_{0:T})$. In the inference problem, however, only $x_{0:T}$ is
made available to the smoothing algorithm. The hidden record $u_{0:T}$ and
the associated pure-state trajectory $\psi_{0:T}$ are inferred jointly. Note that $\psi_t$ is real in this case. 

The hidden $\sigma_z$ measurement tends to localize individual trajectories
toward either $\ket{+z}$ or $\ket{-z}$. In contrast, the observed
$\sigma_x$ measurement does not directly distinguish the signs of the two
$\sigma_z$ branches. Therefore the reconstructed $\av{\sigma_z}\approx 0$, while individual reconstructed trajectories may strongly polarize to $\psi^\dagger \sigma_z \psi\approx \pm 1$. 
This is shown in Fig.~\ref{fig:bimodal-dynamics}. To make the bimodal structure clear, the hidden measurement is chosen to be
stronger than the observed measurement, $\kappa_z>\kappa_x$.
The left plot shows the $\sigma_{x,z}$ components of the true quantum trajectory and their smoothed reconstruction $\av{\sigma_k}_t=\Tr\left(\rho_t^S \sigma_k\right), k=x,z$.  The second and third plots show the $\sigma_{x,z}$ components of the individual posterior sample trajectories and their means. Initially, all particles are concentrated near (z=0),
corresponding to the known initial state $\ket{+x}$. As the hidden
measurement acts, the distribution broadens and subsequently separates into
two modes. At late times, the modes are concentrated near (z=+1) and
(z=-1), representing candidate trajectories that have localized toward
$\ket{+z}$ and $\ket{-z}$, respectively.
The two modes arise because both are compatible with the same observed record and express Bayesian uncertainty about which hidden trajectory occurred.
The right plot shows the so-called localization $\E((\psi_t^\dagger \sigma_z \psi_t)^2\mid x_{0:T})$, which is a fourth-order statistic that quantifies the bimodal localization and that cannot be computed from the density matrix. It shows the strong polarization of the latent quantum states towards $\pm z$.
\begin{figure}
\bc
\includegraphics[height=0.2\textwidth]{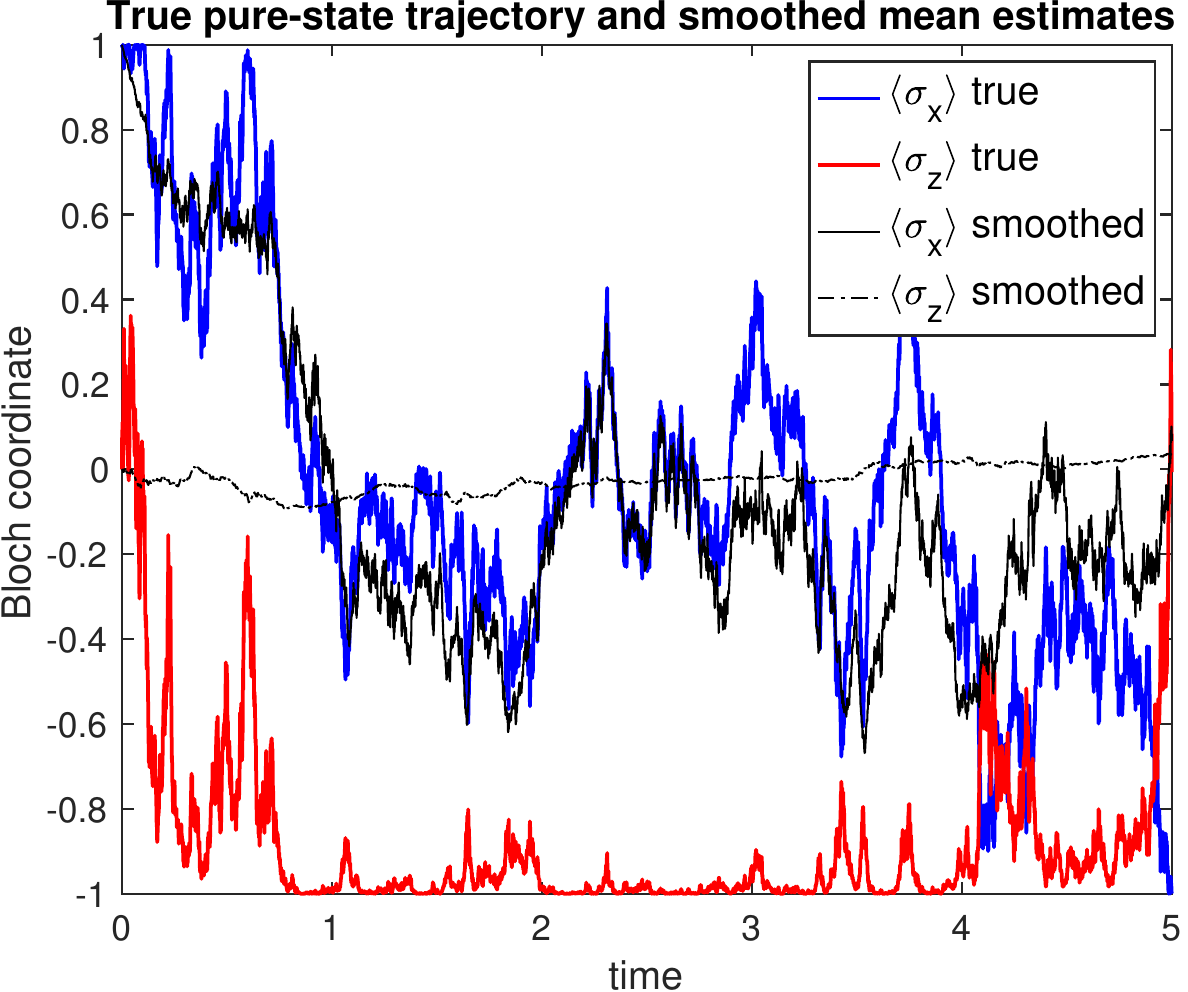}
\includegraphics[height=0.2\textwidth]{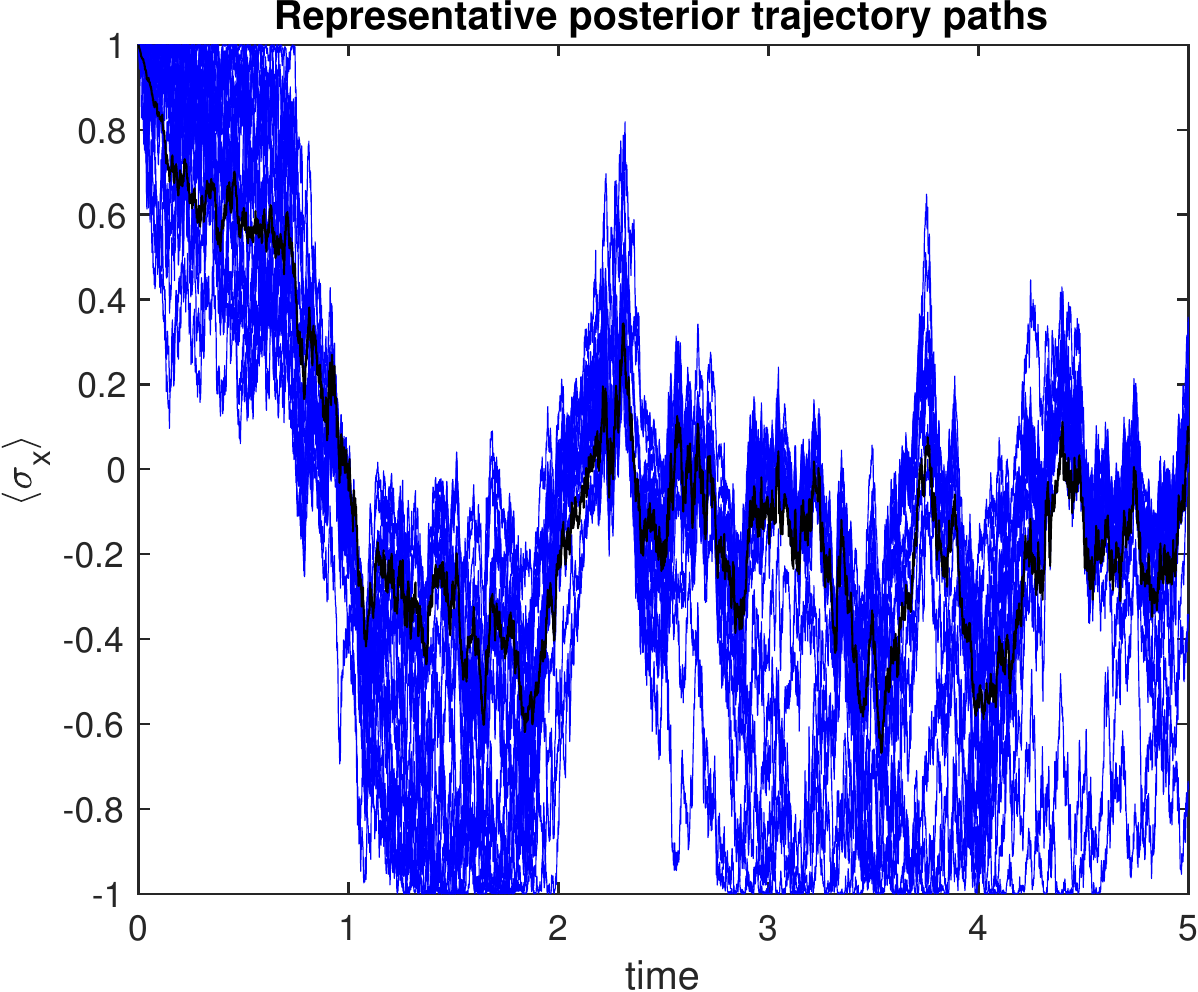}
\includegraphics[height=0.2\textwidth]{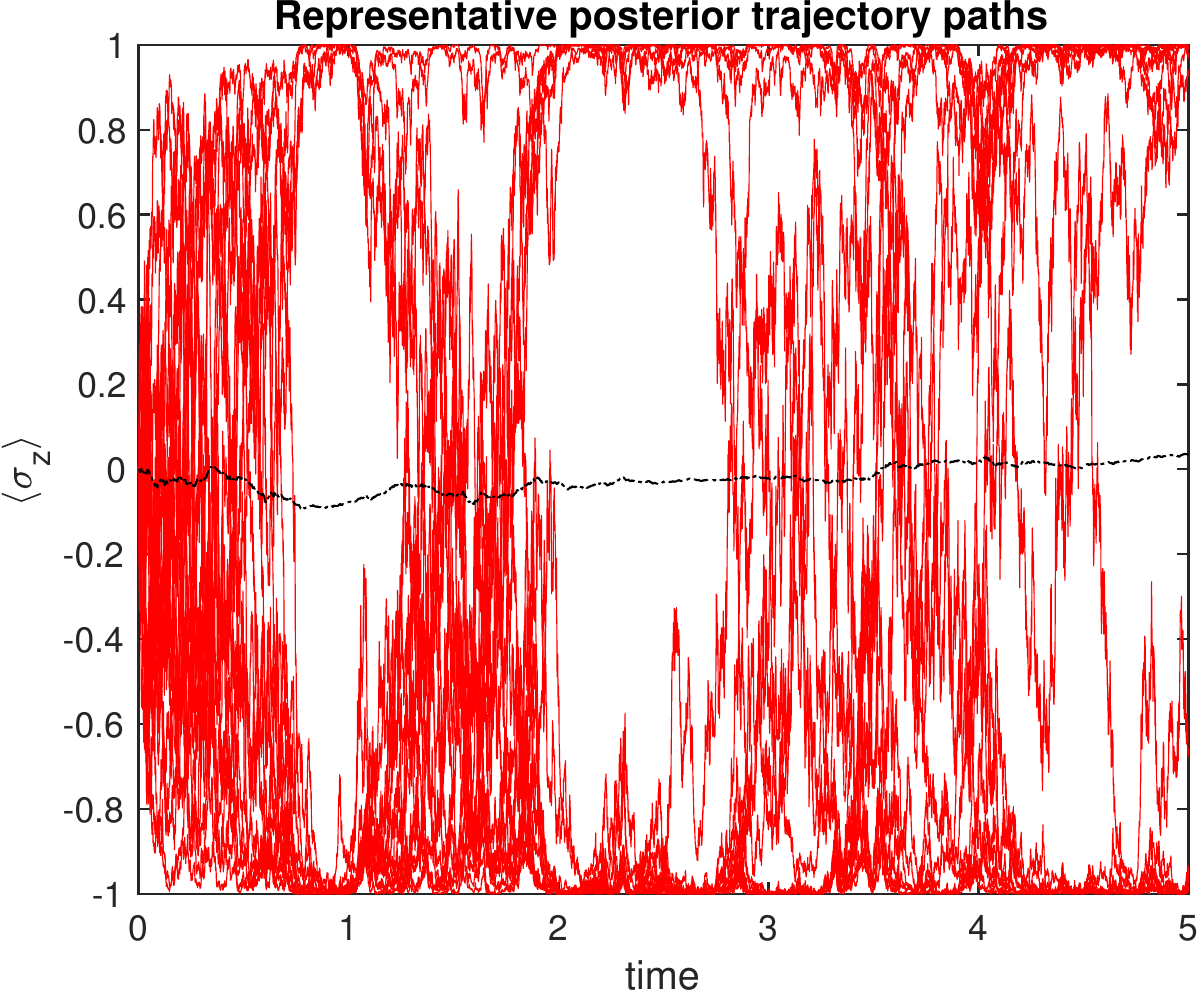}
\includegraphics[height=0.2\textwidth]{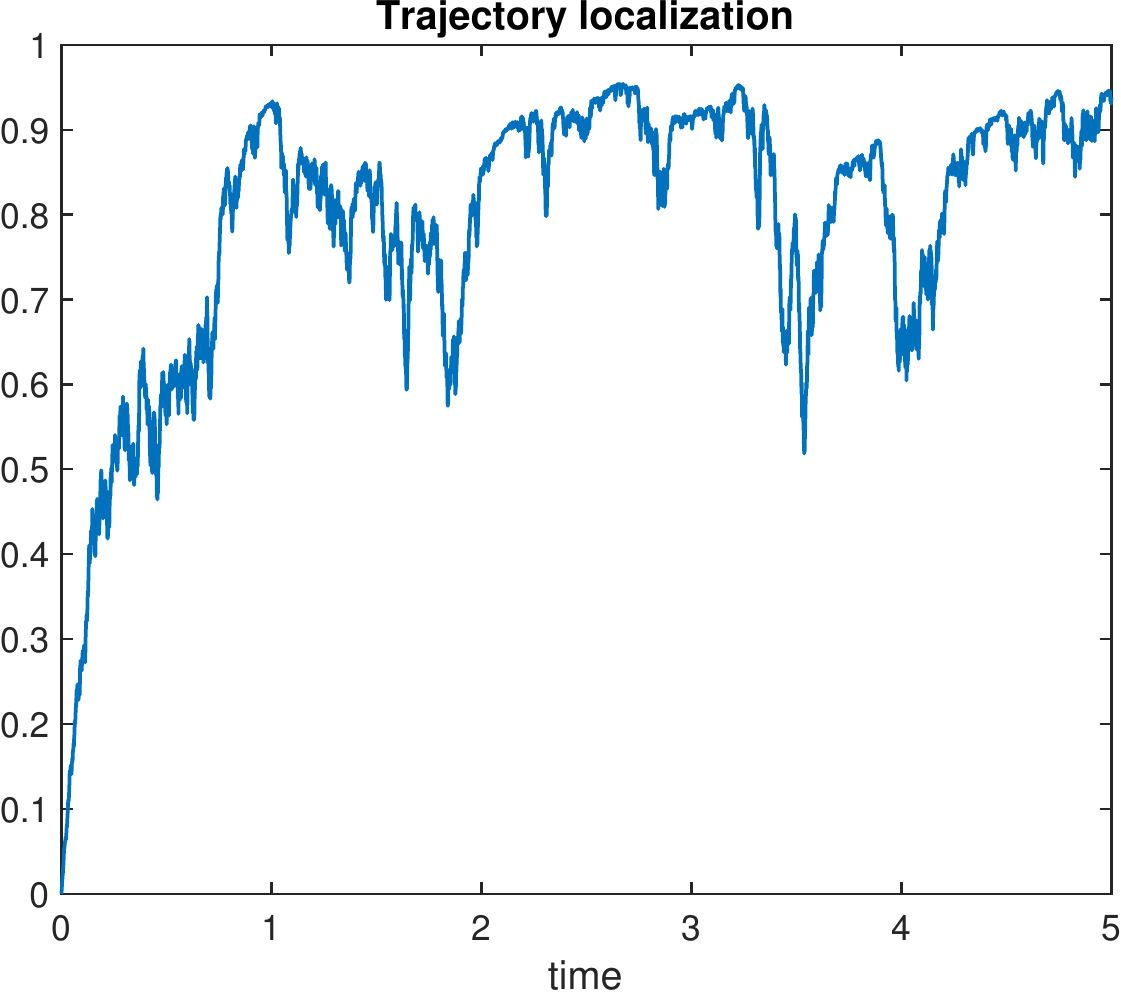}
\ec
\caption{Smoothing by particle filtering. Subplots are numbered 1-4. 1: True pure quantum trajectory and smoothed reconstruction. 
2-3: 20 highest weight particle trajectories approximating the posterior $\cP(\psi_{0:T}\mid x_{0:T})$. 
4: Smoothed localization $\E((\psi_t^\dagger \sigma_z \psi_t)^2\mid x_{0:T})$. $T=5, \kappa_x=0.12,\kappa_z=1.0, dt=0.001, N=500$. %Effective sample size $N_\text{eff}=487$ out of $N=500$ trajectories.
}
\label{fig:bimodal-dynamics}
\end{figure}
%hybrid_quantum_smoothing_bimodal_no_resampling_fixed.m

While according to the postulates of quantum theory, the density matrix (together with the effect matrix as discussed in section~\ref{section:retrograde}) is all that is needed to predict measurement outcomes, 
this experiment demonstrates that
the posterior over quantum states 
\beaa
\cP\left(
\psi_t\mid x_{0:T}
\right)
\quad\text{or}\quad \cP(\psi_{0:T}\mid x_{0:T})
\eeaa
solves a richer inference problem: it also
represents uncertainty about the unobserved environmental record and the
associated latent pure-state history.
The bimodality in this example illustrates the additional
inferential content retained by Bayesian posterior distribution while
remaining fully consistent with the operational role of the density matrix.

\section{Discussion}

We have formulated diffusive hybrid quantum-classical dynamics as an ordinary stochastic process on an enlarged state space containing both quantum and classical variables. Theorem~\ref{thm4} suggests a different interpretation of the universality of Lindblad dynamics. Conventionally, the Lindblad form is obtained by requiring complete positivity, or by embedding the system in a larger unitary quantum dynamics, from which complete positivity follows. Here neither assumption is made. Instead, the quantum state is embedded in an ordinary stochastic quantum--classical dynamics, and the combination of autonomous second-moment evolution, trace preservation, and a nondegenerate coupling between the quantum and classical noise sectors is sufficient to recover the hybrid Lindblad equation. Thus the assumption of underlying unitary quantum dynamics is a sufficient, but not a necessary, route to Lindblad dynamics: the same effective dynamics can arise from a stochastic quantum--classical model. This does not exclude a unitary dilation---indeed, the resulting completely positive dynamics may admit one as a mathematical representation---but shows that such a dilation need not be taken as the physical mechanism underlying the dynamics.

The scope of the present result is limited to the class of diffusive Markovian hybrid dynamics considered in Eqs.~\ref{fp2lindblad1}-\ref{joint}. The derivation assumes that the noise covariances \(C\), \(Q\), and \(\Gamma\) are independent of the quantum state, that the marginal quantum dynamics is trace preserving, and that the quantum--classical noise-correlation matrix \(\Gamma\) has full column rank. The latter condition requires the number of coupled classical noise directions to be at least as large as the number of quantum noise channels and does not apply to the purely quantum case \(\Gamma=0\). The theorem therefore does not constrain rank-deficient hybrid couplings, purely quantum stochastic wave-function models without a classical sector, jump-type dynamics, or general non-Markovian evolutions. Whether analogous closure results hold in these more general settings remains an open question.

\cite{Diosi2014} provides an interesting example from the perspective of the column rank condition on $\Gamma$. His positive but non-completely-positive qubit dynamics involves three independent operator channels acting on a two-dimensional Hilbert space, and Di\'osi observed that the proof of \cite{Bassi2013} fails when the number of independent Lindblad operators $m$ exceeds the Hilbert-space dimension. The present construction is very different because it necessarily involves a classical sector, which the example does not. Nevertheless, there is a suggestive dimensional analogy. In the present construction, complete Lindblad closure follows when the quantum-classical noise-correlation matrix \(\Gamma\) has full column rank, which requires the number \(m\) of quantum-noise channels not to exceed the number \(n\) of coupled classical dimensions. Thus both arguments fail when some quantum-noise directions cannot be independently resolved by the structure used to constrain them, although the relevant dimensions - and therefore the mathematical mechanisms - are different.

The Bayesian formulation naturally displays information in the distribution $\cP(\psi)$ beyond the second order statistics $\rho$, as illustrated in the example in section~\ref{sec:bimodal-example}. When the observed record does not uniquely determine the quantum trajectory (such as in the non-saturated noise case $Q> \Gamma^\dagger C^{-1}\Gamma$ or in the presence of unobserved classical measurements), higher-order moments characterize the residual distribution over latent trajectories. They do not alter the operational predictions encoded in the conditional density matrix, but they quantify dynamical properties of the underlying stochastic realization that are lost under the projection to its second moment. Such information becomes relevant whenever the scientific question concerns hidden environmental records, dynamical histories, switching events, or parameters governing the underlying stochastic process.

A second conceptual implication concerns quantum measurement and smoothing. 
These are often presented as requiring a specifically quantum generalization of classical smoothing, with operator noncommutativity identified as an important obstruction to a direct application of classical smoothing methods \cite{chantasri2021unifying, laverick_quantum_2023}. 
The present formulation suggests that this obstruction is not fundamental when considered at the level of unravelings that are described by a Bayesian 
time series model. The resulting graphical structure (fig.~\ref{hmm}) differs from the usual classical hidden Markov model, as a result of the correlation $\Gamma$ between process and observation noise.
Within this representation, quantum measurement, filtering and smoothing, as well as the difference from their usual classical counterparts, can be described entirely in terms of classical probabilities, without invoking operator noncommutativity. This difference is accounted for by the correlation between process and observation noise.

Technically, this means that conditioning on an observed classical trajectory becomes an ordinary Bayesian inference problem. The filtered density matrix is the second moment of a causal posterior over latent quantum states, while conditioning on the complete observation record gives a smoothed posterior and a corresponding smoothed second moment. The familiar stochastic master equation and retrograde effect operator arise as operator-level representations of the forward and backward Bayesian recursions. In particular, the linear unraveling provides the natural setting for the unnormalized forward filter and its Hilbert-Schmidt adjoint, whereas the normalized physical filter follows by change of measure and normalization (section~\ref{quantumfiltering}). This perspective also clarifies the relation between several previous notions of quantum smoothing as discussed in section~\ref{smoothing}. 

Computationally, this formulation permits the direct use of classical sequential Monte Carlo methods. The simple particle filter and path-space smoother used here are intended primarily as demonstrations; for long records their familiar weight-degeneracy problem calls for resampling, forward-filter backward-simulation, or related particle-smoothing techniques. The essential point is that no specifically quantum modification of these algorithms is required once the hybrid dynamics has been represented as a classical latent stochastic process.

An attractive extension is parameter learning. Unknown Hamiltonian couplings, Lindblad parameters, noise covariances, or initial-state parameters, collectively denoted by $\theta$, can be included in the probabilistic model. The unnormalized filter yields the likelihood ratio $\Tr(\sigma_T)=\Lambda_T(\theta)=\cP^{(1)}_\theta(x_{0:T})/\cP^{(0)}_\theta(x_{0:T})$ so that physical likelihood factorizes as 
\[
\cP^{(1)}_\theta(x_{0:T})=\Lambda_T(\theta) \cP_\theta^{(0)}(x_{0:T})
\]
The likelihood $\cP_\theta^{(0)}(x_{0:T})$ is particularly simple in the present setting: for $\eta=0$ we have $c_a=0$ so that the classical dynamics decouples from the quantum state and reduces to $dx_i =f_i^c(x_,t) dt+dW_i$. Its likelihood is therefore the ordinary classical diffusion likelihood. 
For parameters that affect only the quantum dynamics, or latent classical dynamics coupled to the observations only through the quantum state, $\cP^{(0)}$ is parameter independent and the likelihood can be computed from the unnormalized filter. 
Otherwise, the $\cP_\theta^{(0)}(x_{0:T})$ must be included explicitly.
This likelihood formulation combines naturally with the particle smoothing framework developed above, which supplies samples or sufficient statistics from the posterior over latent trajectories. This suggests direct likelihood optimization as well as expectation-maximization \cite{dempster1977} and related Bayesian parameter-learning methods.

\section*{Acknowledgements}
I would like to thank the Simons Foundation for supporting my sabatical stay at the Flatiron Institute in 2025-2026, during which part of this research was performed. I would like to thank Anirvan Sengupta for encouraging me to write up these results. 

\small{ 
\bibliography{/Users/bertkappen/doc/authors,/Users/bertkappen/doc/zotero}
\bibliographystyle{apalike} }

\appendix 
\section{Proof of Lemma~\ref{lem:rhoJ_dynamics}}
\label{prooflemma1}

Write the stochastic equation for $\psi$ as
\begin{equation*}
    d\psi_j = F_j\,dt + d\chi_j,
    \qquad
    F=f\psi,
    \qquad
    d\chi=g_a\psi\,d\xi_a.
\end{equation*}
with $j=1,\ldots, d$ the components of the wave function. For a general complex diffusion, define the covariance,
pseudo-covariance, and quantum--classical cross covariance by
\beaa
    \mathbb{E}[d\chi_j d\bar\chi_k]
    =
    K_{jk}\,dt,
\qquad
    \mathbb{E}[d\chi_j d\chi_k]
    =
    S_{jk}\,dt,
\qquad
    \mathbb{E}[dW_i d\chi_j]
    =
    R_{ij}\,dt.
\eeaa
In the present case,
\bea
    K_{jk}
    =
    Q_{ab}
    (g_a\psi)_j
    \overline{(g_b\psi)_k},
\qquad
    R_{ij}
    =
    \Gamma_{ia}(g_a\psi)_j.
    \label{eq:KR_definition}
\eea
and bar denotes complex conjugate.
Introduce the Wirtinger derivatives
\[
    \partial_{\psi_j}
    =
    \frac{1}{2}
    \left(
        \frac{\partial}{\partial\operatorname{Re}\psi_j}
        -
        i\frac{\partial}{\partial\operatorname{Im}\psi_j}
    \right),
    \qquad
    \partial_{\bar\psi_j}
    =
    \frac{1}{2}
    \left(
        \frac{\partial}{\partial\operatorname{Re}\psi_j}
        +
        i\frac{\partial}{\partial\operatorname{Im}\psi_j}
    \right).
\]

We denote the backward generator by $\mathcal L_{\rm B}$ and its
formal adjoint, the forward Fokker--Planck operator, by
$\mathcal L_{\rm F}$.
They are related  by
\[
    \int dx\,d\psi\,
    \cP\,\mathcal L_{\rm B}\Phi
    =
    \int dx\,d\psi\,
    \Phi\,\mathcal L_{\rm F}\cP,
    %\label{forwardbackwardFP}
\]
assuming that the boundary terms vanish.
Hence,
\begin{equation}
    -\partial_t\Phi
    =
    \mathcal L_{\rm B}\Phi,
    \qquad
    \partial_t\cP
    =
    \mathcal L_{\rm F}\cP.
\end{equation}
For the dynamics Eq.~\ref{fp2lindblad1}, we obtain 
\begin{align}
    \cL_{\rm B} \Phi
    ={}&
    F_j\partial_{\psi_j}\Phi
    +
    \bar F_j\partial_{\bar\psi_j}\Phi
    +
    h_i\partial_i\Phi
    +
    K_{jk}
    \partial_{\psi_j}\partial_{\bar\psi_k}\Phi
    +
    \frac{1}{2}
    S_{jk}
    \partial_{\psi_j}\partial_{\psi_k}\Phi
    +
    \frac{1}{2}
    \bar S_{jk}
    \partial_{\bar\psi_j}\partial_{\bar\psi_k}\Phi
    \nonumber\\
    &+
    R_{ij}
    \partial_i\partial_{\psi_j}\Phi
    +
    \bar R_{ij}
    \partial_i\partial_{\bar\psi_j}\Phi
    +
    \frac{1}{2}
    C_{ii'}
    \partial_i\partial_{i'}\Phi.
    \label{eq:complex_backward_generator}\\
    \cL_{\rm F}\cP
    ={}&
    -\partial_{\psi_j}
    \left(
        F_j {\cal P}
    \right)
    -
    \partial_{\bar\psi_j}
    \left(
        \bar F_j {\cal P}
    \right)
    -
    \partial_i
    \left(
        h_i {\cal P}
    \right)
    +
    \partial_{\psi_j}
    \partial_{\bar\psi_k}
    \left(
        K_{jk} {\cal P}
    \right)
    +
    \frac{1}{2}
    \partial_{\psi_j}
    \partial_{\psi_k}
    \left(
        S_{jk} {\cal P}
    \right)
    +
    \frac{1}{2}
    \partial_{\bar\psi_j}
    \partial_{\bar\psi_k}
    \left(
        \bar S_{jk} {\cal P}
    \right)
    \nonumber\\
    &+
    \partial_i
    \partial_{\psi_j}
    \left(
        R_{ij} {\cal P}
    \right)
    +
    \partial_i
    \partial_{\bar\psi_j}
    \left(
        \bar R_{ij} {\cal P}
    \right)
    +
    \frac{1}{2}
    \partial_i
    \partial_{i'}
    \left(
        C_{ii'} {\cal P}
    \right).
    \label{eq:complex_fokker_planck}
   \end{align}
Define $\Phi_{kl}=\psi_k\bar\psi_l$. 
Its nonzero Wirtinger derivatives are
\beaa
    \partial_{\psi_j}\Phi_{kl}
    =
    \delta_{jk}\bar\psi_l,
    \qquad
    \partial_{\bar\psi_j}\Phi_{kl}
    =
    \psi_k\delta_{jl},
    \qquad
    \partial_{\psi_j}
    \partial_{\bar\psi_m}\Phi_{kl}
    =
    \delta_{jk}\delta_{ml}.
\eeaa
The purely holomorphic and antiholomorphic second derivatives vanish:
\[
    \partial_{\psi_j}
    \partial_{\psi_m}\Phi_{kl}
    =
    0,
    \qquad
    \partial_{\bar\psi_j}
    \partial_{\bar\psi_m}\Phi_{kl}
    =
    0.
\]
Consequently, the pseudo-covariance $S_{jk}$ does not contribute to
the evolution of $\psi\psi^\dagger$.

Using Eq.~\ref{joint} and~\ref{eq:complex_fokker_planck} and integrating by parts, gives
\begin{align}
    \frac{\partial(\rho_J)_{kl}}{\partial t}
    ={}&
    \int d\psi\,
    \cP(\psi,x,t)
    \left(
        F_k\bar\psi_l
        +
        \psi_k\bar F_l
        +
        K_{kl}
    \right)
    -
    \partial_i
    \int d\psi\,
    \cP(\psi,x,t)
    \left(
        R_{ik}\bar\psi_l
        +
        \psi_k\bar R_{il}
    \right)
    \nonumber\\
    &-
    \partial_i
    \int d\psi\,
    \cP(\psi,x,t)
    A_i\psi_k\bar\psi_l
    +
    \frac{1}{2}
    \partial_i\partial_{i'}
    \int d\psi\,
    \cP(\psi,x,t)
    C_{ii'}\psi_k\bar\psi_l.
    \label{eq:rhoJ_component_proof}
\end{align}

Using $F=f\psi$ and Eqs.~\ref{eq:KR_definition}, we have
\beaa
    F_k\bar\psi_l
    =
    \bigl(f\psi\psi^\dagger\bigr)_{kl},
    \qquad
    \psi_k\bar F_l
    =
    \bigl(\psi\psi^\dagger f^\dagger\bigr)_{kl},
    \qquad
    K_{kl}
    =
    \left(
        Q_{ab}
        g_a\psi\psi^\dagger g_b^\dagger
    \right)_{kl},
    \qquad
    R_{ik}\bar\psi_l
    =
    \left(
        \Gamma_{ia}
        g_a\psi\psi^\dagger
    \right)_{kl}.
\eeaa
Substitution into Eq.~\eqref{eq:rhoJ_component_proof} and returning to
matrix notation yields Eq.~\eqref{fp2lindblad2}.

\section{Proof of Lemma~\ref{lemma3} }
Let $\mathcal H\cong\mathbb C^d$ and let $M_d(\C)$ denote the space of complex $d\times d$ matrices and let
\(\cR:M_d(\mathbb C)\to M_d(\mathbb C)\) be a complex-linear, Hermiticity-preserving map. Suppose that for every nonzero \(\psi\),
\bea
\Pi_\psi^\perp\, \mathcal R(P_\psi)\, \Pi_\psi^\perp=0, \qquad P_\psi=\psi\psi^\dagger\label{eq:tangent}
\eea
where \(\Pi_\psi^\perp\) projects onto the subspace orthogonal to \(\psi\). Then there exists a fixed operator \(B\) such that
\bea
\mathcal R(X)=BX+XB^\dagger\label{lemma3_statement}
\eea
for every \(X\).
\label{appendix:lemma3}

\begin{proof}

Let $\{e_i\}_{i=1}^d$ be an orthonormal basis and write
\[
E_{ij}=e_ie_j^\dagger.
\]

\subsection*{Step 1: basis vectors determine a candidate $B$}

Take $\psi=e_i$. Then $P_\psi=E_{ii}$. Condition \eqref{eq:tangent} implies
\[
e_k^\dagger \cR(E_{ii}) e_\ell=0
\qquad(k,\ell\neq i).
\]
Since $\cR(E_{ii})$ is Hermitian, it therefore must be of the form
\[
\cR(E_{ii})
=
%r_iE_{ii}
%+
%\sum_{k\neq i}
%\left(
%b_{ki}E_{ki}+b_{ki}^*E_{ik}
%\right),
e_i b_i^\dagger + b_i e_i^\dagger
\]
for some arbitrary vector $b_i$. 
%with $r_i\in\mathbb R$.
Define the linear operator $B$ by its operation on the basis vectors as $Be_i =b_i$. 
%Then 
%\[
%B_{ki}=b_{ki}\quad(k\neq i),
%\qquad
%\operatorname{Re}B_{ii}=\frac{r_i}{2}.
%\]
%The imaginary parts of $B_{ii}$ are left free for now. 
Then
\[
\cR(E_{ii})
=
BE_{ii}+E_{ii}B^\dagger
\]
for every $i$. This establishes Eq.~\ref{lemma3_statement} for $X=E_{ii}$. In order to establish the result for all Hermitian matrices $X=\sum_{ij} x_{ij} E_{ij}$, we also need $\cD(E_{ij})=0$ for $i\ne j$. This is established in steps 2 and 3.

Define
\[
\mathcal D(X)
=
\mathcal R(X)-BX-XB^\dagger.
\]
Thus
\[
\mathcal D(E_{ii})=0
\qquad\forall i.
\]
\subsection*{Step 2: two-component superpositions determine pairwise coherences}
Take
\[
\psi=e_i+ze_j,
\qquad i\neq j,\quad z\in\mathbb C.
\]
Then
$
P_\psi
=
E_{ii}+z^*E_{ij}+zE_{ji}+|z|^2E_{jj}.
$
By linearity and $\mathcal D(E_{ii})=\mathcal D(E_{jj})=0$,
\begin{equation}
\mathcal D(P_\psi)
=
z^*A+zA^\dagger,
\qquad
A:=\mathcal D(E_{ij}),
\label{eq:DPA}
\end{equation}
where Hermiticity preservation was used to write
$\mathcal D(E_{ji})=A^\dagger$.
Because Eq.~\ref{eq:tangent} and  $ BP_\psi+P_\psi B^\dagger$ is tangent at $P_\psi$, 
$\mathcal D(P_\psi)$ must also satisfy
\[
\Pi_\psi^\perp\,\mathcal D(P_\psi)\,\Pi_\psi^\perp=0
\]
for every complex $z$.
We now show explicitly that this forces
\[
A=i\omega_{ij}E_{ij}.
\]

First let $e_k,e_\ell$ be basis vectors with
$k,\ell\notin\{i,j\}$.  Since both are orthogonal to $\psi$,
\[
0
=
e_k^\dagger \mathcal D(P_\psi) e_\ell
=
z^*A_{k\ell}+zA_{\ell k}^*
\]
for every $z\in\mathbb C$.  Taking, for example, $z=1$ and $z=i$ shows
$
A_{k\ell}=A_{\ell k}=0.
$
Thus the block of $A$ entirely outside the $i,j$ subspace vanishes.

Next let $k\notin\{i,j\}$ and introduce
$
v_z=-z^*e_i+e_j.
$
Since $v_z\perp \psi$ 
and $e_k\perp\psi$, tangency gives
\[
0=e_k^\dagger \mathcal D(P_\psi) v_z.
\]
Using \eqref{eq:DPA},
\[
0
=
-(z^*)^2A_{ki}
+z^*A_{kj}
-|z|^2A_{ik}^*
+zA_{jk}^*.
\]
This holds for every complex $z$.  The functions
$(z^*)^2$, $z^*$, $|z|^2$, and $z$ are linearly independent as functions
of $z$, and hence
\[
A_{ki}=A_{kj}=A_{ik}=A_{jk}=0.
\]
Therefore $A$ has support only in the two-dimensional subspace
$\operatorname{span}\{e_i,e_j\}$.

Write its remaining block as
\[
A\big|_{\{i,j\}}
=
\begin{pmatrix}
a&b\\
c&d
\end{pmatrix}.
\]
Since $v_z\perp \psi$, tangency also
requires
\[
0
=
\langle v_z|\mathcal D(P_\psi)|v_z\rangle
=
z^*\langle v_z|A|v_z\rangle
+
z\langle v_z|A^\dagger|v_z\rangle.
\]
One easily computes
$
\langle v_z|A|v_z\rangle
=
a|z|^2-bz-cz^*+d,
$
and therefore
\begin{align*}
0=
a z^*|z|^2+a^*z|z|^2
-(b+b^*)|z|^2 
-c(z^*)^2-c^*z^2
+d z^*+d^*z .
\end{align*}
Since this identity holds for every $z\in\mathbb C$, the coefficients of
the independent powers and phases of $z$ must vanish.  Hence
\[
a=0,\qquad
c=0,\qquad
d=0,\qquad
b+b^*=0.
\]
Thus $b$ is purely imaginary.  Writing
\[
b=i\omega_{ij},
\qquad \omega_{ij}\in\mathbb R,
\]
we obtain
\begin{equation}
\mathcal D(E_{ij})
=
i\omega_{ij}E_{ij}.
\label{eq:D_Eij}
\end{equation}
Hermiticity preservation gives
$
\mathcal D(E_{ji})
=
\mathcal D(E_{ij})^\dagger
=
-i\omega_{ij}E_{ji},
$
so
\[
\omega_{ji}=-\omega_{ij}.
\]
Thus two-component superpositions determine the action on each coherence
$E_{ij}$ up to a pairwise phase rate $\omega_{ij}$.

\subsection*{Step 3: three-component superpositions enforce consistency}

Take three distinct indices $i,j,k$ and
\[
\psi=e_i+e_j+e_k.
\]
Using the result of Step 2,
\[
\mathcal D(P_\psi)
=
i\omega_{ij}E_{ij}
+i\omega_{jk}E_{jk}
+i\omega_{ki}E_{ki}
+\text{h.c.}
\]
where we used $\cD(E_{ii})=\cD(E_{jj})=\cD(E_{kk})=0$. 
Choose two vectors orthogonal to $\psi$:
\[
u=e_i-e_j,
\qquad
v=e_j-e_k.
\]
Tangency implies
\[
0=u^\dagger \mathcal D(P_\psi)v
=
i\left(
\omega_{ij}+\omega_{jk}+\omega_{ki}
\right),
\]
hence
\[
\omega_{ij}+\omega_{jk}+\omega_{ki}=0.
\]
Define $\theta_i =\omega_{i1}$ for all $i$. Then with $k=1$ we find \footnote{
Note, that this result does requires three distinct indices $i,j,k$ and thus does not cover the case $d=1,2$.When $d=1$, $\omega_{11}=0$. When $d=2$, $\omega_{12} $ can always be written as $\theta_1 -\theta_2$. }
\[
\omega_{ij}=\theta_i-\theta_j.
\]

Define the operator $K = i\sum_i \theta_i E_{ii}$. Then $KE_{ij}+E_{ij}K^\dagger =i\theta_i E_{ij}-i\theta_j E_{ij}=i\omega_{ij}E_{ij}$. 
Using Eq.~\ref{eq:D_Eij} and the definition of $\cD(X)$ we also have 
\beaa
\cR(E_{ij})=BE_{ij} +E_{ij}B^\dagger +i\omega_{ij}E_{ij}=(B+K) E_{ij}+E_{ij} (B+K)^\dagger
\eeaa
Redefine $B+K \to B$ and noting that any Hermitian matrix can be written as $X=\sum_{ij} x_{ij} E_{ij}$ this establishes the result.
\end{proof}
%\section*{Why the three steps are needed}
%
%The structure is:
%
%\begin{enumerate}
%\item $\psi=e_i$ fixes the action on diagonal rank-one directions and
%constructs the columns of a candidate $B$, except for imaginary diagonal
%entries.
%
%\item $\psi=e_i+ze_j$ introduces $E_{ij}$ and $E_{ji}$ through
%\[
%P_\psi
%=
%E_{ii}+z^*E_{ij}+zE_{ji}+|z|^2E_{jj}.
%\]
%Linearity therefore allows the off-diagonal action to be constrained by
%polarization.  Only a pairwise phase ambiguity survives.
%
%\item $\psi=e_i+e_j+e_k$ tests whether the phase ambiguities for the three
%pairs $(i,j)$, $(j,k)$, and $(k,i)$ are globally compatible.  Tangency
%forces
%\[
%\omega_{ij}+\omega_{jk}+\omega_{ki}=0,
%\]
%which means
%\[
%\omega_{ij}=\theta_i-\theta_j.
%\]
%Those phases are exactly the freedom in the imaginary diagonal part of one
%fixed operator $B$.
%\end{enumerate}

\section{Different noise choices}

\label{noisechoice}
The noise $d\xi_a=d\xi_{a,r} + id\xi_{a,i}$ is complex valued with possibly independent or correlated real and imaginary components. 
Here we will discuss how different choices of the complex noise affect the hybrid dynamics Eq.~\ref{fp2lindblad3} and its unravelings Eq.~\ref{sse3} and~\ref{dx}. 
The hybrid dynamics Eq.~\ref{fp2lindblad3} depends on the noise through the covariances $Q_{ab}dt=\av{d\xi_a d\xi_b^*}$ and $\Gamma_{ia}dt=\av{dW_i d\xi_a}$.
We may assume without loss of generality that $Q$ is real diagonal and that $\Gamma_{ia}$ is real
\footnote{Consider 
the linear transformation 
\beaa
\tilde L_c = R_{ca}^* L_a\qquad d\xi_a=R^*_{ca}d\tilde \xi_c
\eeaa
with $R_{ca}\in \C$. 
This implies $Q = R^\dagger \tilde Q R$ with $\tilde Q_{cd}=\av{d\tilde \xi_c d\tilde \xi_d^\dagger}$, $\Gamma_{ia} = R^*_{ca}\tilde \Gamma_{ic}$ with $\tilde \Gamma_{ic}=\av{dW_i d\tilde\xi_c}$ and $\tilde c_c =R^*_{ca}c_a$. Since $Q_{ab} L_a L_b^\dagger = \tilde Q_{ab} \tilde L_a \tilde L_b^\dagger, (L_a -c_a)d\xi_a = (\tilde L_a -\tilde c_a) d\tilde \xi_a$ and $\Gamma_{ia} L_a = \tilde \Gamma_{ia} \tilde L_a$, this transformation leaves the unravelings Eq.~\ref{sse3} and~\ref{dx}, and therefore also the joint density dynamics Eq.~\ref{fp2lindblad3} invariant. 
When $R=DU$, with $U$ the unitary transformation that transform $Q$ to its eigenbasis and $D$ an arbitrary diagonal matrix, $\tilde Q$ is real and diagonal for any $D$ and $\tilde \Gamma_{ib}= \Gamma_{ia} U_{ba} D_b$ ($\Gamma = \tilde \Gamma R^*=\tilde \Gamma D^* U^*$ implies $\Gamma^* U^\dagger D^\dagger = \tilde \Gamma^*$) so that we can choose $D$ such that $\tilde \Gamma$ is real. We can thus assume $Q$ real and diagonal and $\Gamma$ real without loss of generality. }.
There are different choices for $d\xi_a$ with identical $Q_a$ and $\Gamma_{ia}$ and thus identical hybrid dynamics Eq.~\ref{fp2lindblad3}. Since $\Gamma_{ia}dt =\av{dW_i d\xi_a}$ is real, this implies that either $d\xi_a$ is real or $d\xi_{a,i}$ is uncorrelated to $dW_i$. 
We consider these two choices
\bi
\item {\bf Real noise}. We choose $d\xi_a$ real and $\av{d\xi_a^2}=Q_a$. Then $\Gamma_{ia} = \av{dW_i d\xi_a}$ is real. 
The noise term $(L_a -c_a) P d\xi_a +h.c.$ in the dynamics Eq.~\ref{sme3} for an observable $\Tr (AP)$ is given as
\beaa
d\chi_a &=& (G_a +h.c.)d\xi_a \qquad G_a = \Tr(A (L_a-c_a) P)\\
\av{d\chi_a^2}&=& \left(G_a+h.c.\right)^2 Q_a
\eeaa
The coupling to $x_i$ is  $\Gamma_{ia} (c_a +h.c)$. 
\item {\bf Circle symmetric noise}. A common choice in the unraveling literature (for instance  \cite{diosi_hybrid_2023}) is to choose $d\xi_a$ complex valued 'circle symmetric', which means  $\av{d\xi_{a,r}^2}=\av{d\xi_{a,i}^2}=\frac{Q_a}{2}dt, \av{d\xi_{a,r} d\xi_{a,i}}=0$. Define $\Gamma_{ia} dt=\av{dW_i d\xi_a}=\Gamma_{ia,r}dt+i\Gamma_{ia,i}dt$. It is real-valued when $\av{dW_i d\xi_{a,i}}=0$, i.e. $d\xi_{a,i}$ is uncorrelated to $dW_i$. 
The noise term $(L_a -c_a) P d\xi_a +h.c.$ in the dynamics Eq.~\ref{sme3} for an observable $\Tr (AP)$ is given as
\beaa
d\chi_a &=& G_a d\xi_a +h.c. \qquad G_a = \Tr(A (L_a -c_a)P)\\
\av{d\chi_a^2}&=& 2|G_a|^2 Q_a 
\eeaa
The coupling to $x$ is  $\Gamma_{ia} (c_a +h.c)$. 
\ei
Thus, the two noise choices yield identical hybrid dynamics and differ in the noise of the unraveling of the quantum dynamics. In addition, there is a difference in the maximal size of $\Gamma_{ia}$, which we illustrate in the case of a single
Lindblad operator and a single measurement and $C,Q,\Gamma$ real scalars. In the case of real noise $\Gamma$ must be chosen such that the covariance matrix of $dW, d\xi$ is
\beaa
D=
%\left(\begin{tabular}{ccc} $ C$ & $\Gamma_r$ &$\Gamma_i$\\ $\Gamma_r$ & $M_{rr}$ & $M_{ri}$\\
%$\Gamma_i$ & $M_{ir}$ & $M_{ii}$
 %\end{tabular} \right)=
 \left(\begin{tabular}{cc} $ C$ & $\Gamma$ \\ $\Gamma$ & $Q$
 \end{tabular} \right)\ge 0
\eeaa
which implies $\Gamma^2 \le QC$. In the case of circle symmetric noise, $\Gamma$ must be chosen such that the covariance matrix of $dW, d\xi_r,d\xi_i$ is
\beaa
D=
%\left(\begin{tabular}{ccc} $ C$ & $\Gamma_r$ &$\Gamma_i$\\ $\Gamma_r$ & $M_{rr}$ & $M_{ri}$\\
%$\Gamma_i$ & $M_{ir}$ & $M_{ii}$
 %\end{tabular} \right)=
 \left(\begin{tabular}{ccc} $ C$ & $\Gamma$ & $0$\\ $\Gamma$ & $\frac{Q}{2}$ & $0$\\
 $0$ & $0$ & $\frac{Q}{2}$
 \end{tabular} \right)\ge 0
\eeaa
which implies $\Gamma^2 \le \frac{QC}{2}$. 
Thus real noise unravelings allow for a stronger coupling between the quantum and classical variable than the circle symmetric complex noise unraveling. Since the imaginary noise is uncorrelated to the classical noise, the use of circle symmetric noise does not seem natural for unravelings of the hybrid dynamics. We therefore will continue with real valued noise. 

\section{The Kushner-Stratonovich equation}
\label{section:KS}
Let \(z_t \in \R^d\) denote the hidden state and let \(x_t\in \R^n\) denote the
observed process. We consider the stochastic dynamics
\beaa
 dz_t&=& a(z_t,x_t,t)dt + B(z_t,x_t,t)d\xi_t, \\
 dx_t&=& h(z_t,x_t,t)dt + dW_t
\eeaa
with $z_t, a\in \R^d, \xi_t\in \R^k, B\in \R^{d\times k}, x_t, h, dW_t \in \R^n$ and
with noise covariance
\beaa
 \left<dW_i dW_j\right> &= C_{ij}dt, \qquad  \left<d\xi_a d\xi_b\right> = Q_{ab}dt\qquad
 \left<dW_i d\xi_a\right> = \Gamma_{ia}dt .
\eeaa
In this section we derive the Kushner--Stratonovich (KS) equation which describes the dynamical evolution of the filtered estimate $\pi_t(z_t|x_{0:t})$. The usual derivation of the KS equation assumes \(\Gamma=0\). Here, we keep
\(\Gamma\) arbitrary which is essential for the hybrid quantum-classical dynamics. 
Subsequently, we apply the KS equation to the unraveling dynamics Eq.~\ref{sse3} and~\ref{dx} and 
obtain an equation for the dynamical evolution of the filtered quantum density matrix.

Let \(\phi(z)\) be a smooth test function. By It\^o calculus,
\begin{align}
 d\phi
 =
 {\cal L}\phi\,dt
 + D_a \phi d\xi_a ,\label{dphi}
\end{align}
where $\cL\phi= a\cdot \nabla \phi  +\frac12 \Tr\left(BQB^T \nabla^2 \phi\right)$  is the backward generator of the hidden process,
with \(x_t\) considered as a given time-dependent input  and $D_a \phi =(B_{\cdot a}\cdot \nabla)\phi$. 
Define the filtered expectation
\begin{align}
 \pi_t(\phi)
 =
 \E_t [\phi \mid x_{0:t}]=\int dz \phi(z) \pi_t(z\mid x_{0:t})\label{pi_definition}
\end{align}
The change in $\pi_t(\phi)$ in a small time step $dt$ is due to the change in $\phi$ Eq.~\ref{dphi} and due to the change in $\E_t(\phi_{t+dt}\mid x_{0:t})$ to $\E_{t+dt}(\phi_{t+dt}\mid x_{0:t},dx)$ as a result of the incremented observation $dx$. The first contribution is given by
\bea
\E_t(\phi_{t+dt}\mid x_{0:t}) = \E_t(\phi_t+d\phi\mid x_{0:t})=\pi_t(\phi)+\pi_t(\cL \phi)dt \label{first}
\eea
For the second, we observe that the relation between $\E_t(\phi_{t+dt} \mid x_{0:t})$ and $\E_{t+dt}(\phi_{t+dt} \mid x_{0:t},dx)$ can be obtained by noting that to first order in dt, the update is determined by the conditional covariance with the observation innovation
\footnote{\label{gaussian_conditional}
Denote $x_a=\phi_{t+dt}, x_b=dx$ two stochastic variables that are Gaussian distributed with 
 means $\mu_a,\mu_b$ and covariance matrix $\Sigma=\left(\begin{matrix} \Sigma_{aa} & \Sigma_{ab} \\ \Sigma_{ba} & \Sigma_{bb} \end{matrix}\right)$.
Then
\beaa
p(x_a|x_b) &=& \cN(x_a| \mu_{a|b}, \Sigma_{a|b})\\
\Sigma_{a|b}&=&\Sigma_{aa} - \Sigma_{ab} \Sigma_{bb}^{-1}\Sigma_{ba}\qquad
\mu_{a|b}=\mu_a+\Sigma_{ab}\Sigma_{bb}^{-1}(x_b - \mu_b)
\eeaa
Since all quantities are  conditioned in $x_{0:t}$, the result follows by noting that $\mu_a =\E(\phi_{t+dt}\mid x_{0:t}), \Sigma_{ab} =\text{cov}_t(\phi_{t+dt}, dx), \Sigma_{bb}=\text{cov}_t(dx,dx)$ and $x_b-\mu_b= dx- \pi(h)dt$. 
}
\bea
\E_{t+dt}(\phi_{t+dt}\mid x_{0:t},dx)=\E_t(\phi_{t+dt}\mid x_{0:t}) + \text{cov}_t(\phi_{t+dt}, dx_i) \text{cov}_t(dx_i,dx_j)^{-1}dI_j\label{second}
\eea
where we define the innovation process $dI_i=dx_i-\pi_t(h_i)dt$. Note that $\text{cov}_t(dx_i,dx_j) =C_{ij}dt$ and compute
\bea
\text{cov}_t(\phi_{t+dt}, dx_i)&=& \text{cov}_t(\phi_t+d\phi_t ,h_i dt+dW_i)=
 \text{cov}_t(\phi_t,h_i dt) +\text{cov}_t(d\phi_t,dW_i)+o(dt)\nonumber\\
 &=&[\pi_t(\phi h_i) -\pi_t(\phi)\pi_t(h_i)]dt  +\pi_t(d\phi_t dW_i)+o(dt)\label{third}
\eea
Combining Eqs.~\ref{first}, \ref{second} and~\ref{third}, the KS equation is
\begin{align}
\boxed{
 d\pi_t(\phi)
 =
 \pi_t({\cal L}\phi)dt
 +
 \left[
   \pi_t(\phi h_i)
   -\pi_t(\phi)\pi_t(h_i)
   +\frac{1}{dt}\pi_t(d\phi_t dW_i)
 \right]
 C^{-1}_{ij}dI_j .
}
\label{eq:KS-correlated}
\end{align}

The first two terms in square brackets give the usual KS gain when $\Gamma=0$. The last
term is due to the quadratic covariation between the hidden process and the
observation process.
%Equivalently, if \(p_t(z)=p(z_t\mid x_{0:t})\) denotes the filtered density,
%then Eq.~\eqref{eq:KS-correlated} can be written in the strong form
%\begin{align}
% dp_t(z)
% =
% {\cal L}^\dagger p_t(z)dt
% +
% \left[
%   \bigl(h_i(z)-\pi_t(h_i)\bigr)p_t(z)
%   -\partial_\mu\bigl(\Gamma_{i\alpha}b_{\mu\alpha}(z)p_t(z)\bigr)
% \right]
% C^{-1}_{ij}dI_j .
%\end{align}
%The derivative term is the strong-form representation of
%\(\pi_t(M_i\phi)\) after integration by parts.

\section{Proof of Lemma~\ref{lemma:linear}}
\label{appendixlemma3}

The linear filter dynamics Eq.~\ref{eq:rho-filter} for $\eta=0$ satisfies 
\begin{equation}
d\sigma_t
=
\cL_L(\sigma_t)\,dt
+
\mathcal K_{i,t}(\sigma_t)
C^{-1}_{ij}dY_j,
\label{eq:app-linear-sigma}
\end{equation}
where
\begin{equation}
dY_j:=dx_j-f_j^c\,dt,
\qquad 
\mathcal K_{i,t}(A)
=
\Gamma_{ia}L_{a,t}A
+ h.c.
\nonumber
\end{equation}
and $\cL_L$ is given by Eq.~\ref{lindblad}.

Define 
\begin{equation}\rho_t:=\frac{\sigma_t}{Z_t}\qquad 
Z_t:=\Tr\sigma_t
\label{Zt}
\end{equation}
We will show that $\rho_t$ satisfies Eq.~\ref{eq:rho-filter} for the norm-preserving unraveling $\eta=1$. We assume \(Z_t>0\), as is the case whenever the realized record has nonzero likelihood under the linear model.
Define
\begin{equation}
m_{i,t}
:=
\Gamma_{ia}\Tr(L_{a,t}\rho_t)
+\mathrm{h.c.}
\nonumber
\end{equation}
Because \(\mathcal L_L\) is trace preserving, taking the trace of Eq.~\eqref{eq:app-linear-sigma} gives
\begin{equation}
dZ_t
=
Z_t m_{i,t}C^{-1}_{ij}dY_j\qquad Z_0=1
\label{eq:app-dZ}
\end{equation}

We now apply It\^o's formula to \(Z_t^{-1}\):
\be
d(Z^{-1})
=
-Z^{-2}dZ+Z^{-3}(dZ)^2=
-Z_t^{-1}m_{i,t}C^{-1}_{ij}dY_j
+
Z_t^{-1}m_{i,t}C^{-1}_{ij}m_{j,t}\,dt.\label{eq:app-inverse-Z}
\ee
where we used $(dZ_t)^2=Z_t^2 m_{i,t}C^{-1}_{ij}m_{j,t}\,dt$.

Using \(\rho_t=Z_t^{-1}\sigma_t\), the It\^o product rule yields
\begin{equation}
d\rho_t
=
Z_t^{-1}d\sigma_t
+
\sigma_t\,d(Z_t^{-1})
+
d\sigma_t\,d(Z_t^{-1}).
\label{eq:app-product}
\end{equation}
The quadratic covariation term is
\begin{align}
d\sigma_t\,d(Z_t^{-1})
&=
-
\mathcal K_{i,t}(\rho_t)
C^{-1}_{ij}
m_{j,t}\,dt.
\label{eq:app-cross-term}
\end{align}
Substituting Eqs.~\eqref{eq:app-linear-sigma}, \eqref{eq:app-inverse-Z}, and \eqref{eq:app-cross-term} into Eq.~\eqref{eq:app-product}, and using the linearity of \(\cL_L\) and \(\mathcal K_{i,t}\), gives
\begin{equation}
d\rho_t
=
\mathcal L_L(\rho_t)\,dt
+
\left[
\mathcal K_{i,t}(\rho_t)
-
m_{i,t}\rho_t
\right]
C^{-1}_{ij}
\left[
dY_j-m_{j,t}\,dt
\right].
\label{eq:app-normalized-general}
\end{equation}
Since $
\bar c_{a,t}=\Tr(L_{a,t}\rho_t)$, 
$m_{i,t}
=
\Gamma_{ia}\bar c_{a,t}
+\mathrm{h.c.}
$ and 
$
\mathcal K_{i,t}(\rho_t)-m_{i,t}\rho_t
=
\Gamma_{ia}
\left(
L_{a,t}-\bar c_{a,t}
\right)\rho_t
+\mathrm{h.c.}
$, so that Eq.~\ref{eq:app-normalized-general} is equal to  Eq.~\ref{eq:rho-filter} for the norm-preserving choice.
In other words, $\rho_t$ defined by Eq.~\ref{eq:rho-sigma-section} satisfies Eq.~\ref{eq:rho-filter}.

To demonstrate Eq.~\ref{eq:TRsigma}, 
let $\cP^{(0)}(x_{0:t})$ and $\cP^{(1)}(x_{0:t})$ denote the likelihood densities of the observed record under the linear and normalized unravelings. Their likelihood ratio is 
\[
\Lambda_t = \frac{\cP^{(1)}(x_{0:t})}{\cP^{(0)}(x_{0:t})}
\]
Under the linear unraveling the observation increment has drift $f^c_i$, whereas under the normalized unraveling its drift is $f_i^c+m_{i,t}$, with the same covariance $C$. Hence the ratio of their one-step Gaussian likelihoods is
\beaa
\frac{\cP^{(1)}(dx_t|x_{0:t})}{\cP^{(0)}(dx_t|x_{0:t})}=\exp\left[
 m_{i,t}C^{-1}_{ij}dY_j
-\frac{1}{2}m_{i,t}C^{-1}_{ij}m_{j,t}\,dt
\right].
\eeaa
Because path likelihoods factorize one step at the time, 
\beaa
\Lambda_{t+dt} =\Lambda_t \frac{\cP^{(1)}(dx_t|x_{0:t})}{\cP^{(0)}(dx_t|x_{0:t})}=\Lambda_t e^{d\lambda_t}
\eeaa
where we define $d\lambda_t = m_i C^{-1}_{ij} dY_j -\frac12 m_i C_{ij}^{-1}m_j dt$.
Expanding to It\^o order $e^{d\lambda_t} = 1 +d\lambda_t + \frac12(d\lambda_t)^2 +o(dt)$ we obtain
\[
e^{d\lambda_t}=1+m_i C_{ij}^{-1}dY_j +o(dt)
\]
Thus
\bea
d\Lambda_t = \Lambda_{t+dt}-\Lambda_t = \Lambda_t m_i C_{ij}^{-1} dY_j
\eea
Equivalently, with $\lambda_t =\log \Lambda_t$, the future likelihood ratio introduced in Appendix~\ref{appendix:effect} satisfies $\ell_t =\Lambda_T/\Lambda_t = \exp(\lambda_T-\lambda_t)$. 
Since $Z_t=\Tr(\sigma_t)$ and $\Lambda_t$ satisfy the same SDE and $Z_0=\Lambda_0=1$, we obtain $\Lambda_t =Z_t$ which proves 
Eq.~\ref{eq:TRsigma}.

\section{The effect operator}
\label{appendix:effect}
There exists a quantum analog of the backward messages that we define below in Eq.~\ref{bmessage}, known as the retrograde filter or effect operator $E_t$ \cite{tsang_time-symmetric_2009}. The retrograde filter equation can be derived for the linear unraveling where we can use the Hilbert-Schmidt adjoint for linear operators. Define the linear operator that describes the filter dynamics of the linear unraveling $\sigma_{t+dt} =\cM_{t,dx_t} (\sigma_t)$. The effect operator satisfies the adjoint dynamics 
\bea
E_t(x) =\cM^\dagger_{t,dx_t}(E_{t+dt}(x'))\qquad E_T=I\label{eq:effect-recursion-section}
\eea
with $x'=x+dx_t$ and the adjoint operator is defined through the relation $\Tr\left(\cM^\dagger_{t,dx_t}(A)B\right)=\Tr\left(A\cM_{t,dx_t}(B)\right)$ for all operators $A,B$. The explicit form of Eq.~\ref{eq:effect-recursion-section} and the adjoint operators is given by Eq.~\ref{explicit}. 

We can relate the effect operator to the Bayesian backward message as follows. 
Define the backward message of the non-linear process ($\eta=1$) in the usual Bayesian sense:
\beaa
b_t(\phi_t,x_t) = \cP^{(1)}(x_{t:T}\mid \phi_t,x_t)\qquad b_T(\phi,x_T)=1
\eeaa
where $\phi$ denotes a normalized wave function. In Appendix~\ref{sec:retrograde-filter} we show that
the effect operator is related to the  backward message as
\bea
b_t(\phi,x)=q_t \Tr\left(E_t(x)P_{\phi}\right) \qquad q_t=\cP^{(0)}(x_{t:T}|x_{0:t})\label{bE}
\eea
Eq.~\ref{bE} thus establishes the relation between the classical notion of the Bayesian backward message
$b_t$ and the quantum notion of the retrograde effect operator $E_t$. The retrograde equation
Eq.~\ref{eq:effect-recursion-section} coincides with the result of \cite{gammelmark_past_2013} who consider the linear unraveling only. 

Because the dynamics of $\sigma_t$ and $E_t$ are related by the adjoint transformation, it is easy to see that 
$\Tr(E_t \sigma_t)$ is independent of $t$. 
Defining (see Eq.~\ref{eq:TRsigma}) 
\beaa
\Lambda_t = \frac{\cP^{(1)}(x_{0:t})}{\cP^{(0)}(x_{0:t})} \qquad \ell_t = \frac{\cP^{(1)}(x_{t:T}\mid x_{0:t})}{\cP^{(0)}(x_{t:T}\mid x_{0:t})}
\eeaa
we obtain (see Appendix~\ref{sec:retrograde-filter})
\beaa
\Tr(E_t \sigma_t) =\Lambda_t \ell_t = \Lambda_T\qquad \Tr(E_t \rho_t) = \frac{\Tr(E_t\sigma_t)}{\Lambda_t} = \ell_t
\eeaa
with  $\Lambda_0=\ell_T=1$.
Thus $(\sigma_t,E_t)$ provides the linear-adjoint operator representation of the classical forward-backward construction. Their contraction gives the complete record likelihood, while $\Tr(\sigma_t) =\Lambda_t$ and $\Tr(E_t\rho_t)=\ell_t$ give the past and future  factors.

\subsection{Relation between effect operator and backward messages}
\label{sec:retrograde-filter}
We derive the retrograde filter equation governed by the effect operator $E_t$. As pioneered by \cite{gammelmark_past_2013} this can be done for the linear unraveling where we can use the Hilbert-Schmidt adjoint for linear operators. First, we connect the backward messages of the linear and non-linear processes. 

Define the backward message of the non-linear process $\eta=1$ in the usual Bayesian sense:
\bea
b_t(\phi_t,x_t) = \cP^{(1)}(x_{t:T}\mid \phi_t,x_t)\qquad b_T(\phi,x_T)=1\label{bmessage}
\eea
where $x_{t:T}=\{dx_s : t\le s < T\}$ denotes the future record of increments, while $x_t$ denotes the current observed state and $\phi_t$ denotes a normalized wave function. The Markov property on the joint state space $(\phi,x)$ and the observed $dx_t$
gives the backward Chapman-Kolmogorov recursion 
\bea
b_t(\phi,x)= \int  \cP^{(1)}(d\phi',dx_t\mid \phi,x) b_{t+dt}(\phi',x')\label{brecursion}
\eea
with $x'=x+dx_t$. 

The one-step normalized process and unnormalized process are related by the Doob transform as we show now.
For the linear unraveling, the stochastic norm
$$
Z_t=\|\psi_t\|^2
$$
has no drift,
$$
dZ_t=\psi_t^\dagger L_a\psi_t\,d\xi_{a,t}+\mathrm{h.c.}
$$
Since the noise increments have zero mean conditional on the present state, this implies
$$
\mathbb E^{(0)}[Z_{t+dt}\mid \psi_t=\psi,x_t=x]=Z_t.
$$
Equivalently, in terms of the one-step transition density of the linear process,
$$
\int d\psi'\,dx\,
P^{(0)}(\psi',dx\mid\psi,x)\,
\frac{\|\psi'\|^2}{\|\psi\|^2}
=1.
$$
Thus the factor \(\|\psi'\|^2/\|\psi\|^2\) can be used to reweight the transitions of the linear process while preserving normalization. Define
$$
\widetilde P^{(1)}(\psi',dx\mid\psi,x)
=
\frac{\|\psi'\|^2}{\|\psi\|^2}
P^{(0)}(\psi',dx\mid\psi,x).
$$
Thus $\widetilde P^{(1)}$ is the Doob transform of the linear transition kernel with \(h(\psi)=\|\psi\|^2\).
The physical (normalized) unraveling is obtained in addition by normalizing the wave function,
$$
\phi'=\hat\psi'=\frac{\psi'}{\|\psi'\|}.
$$
Therefore, for any function \(F\) of the normalized state, the one-step normalized process and unnormalized process are related by the Doob transform
\bea
\int d\phi'\,
F(\phi')P^{(1)}(\phi',dx\mid\phi,x)
=
\int d\psi'\,
F(\hat\psi')
\frac{\|\psi'\|^2}{\|\psi\|^2}
P^{(0)}(\psi',dx\mid\psi,x),
\qquad
\phi=\hat\psi .\label{doob_transform}
\eea

Define the backward messages for the unnormalized process $\beta_t(\psi,x) = \|\psi\|^2 b_t(\hat\psi,x)$. Then the recursion for $b_t$ and the Doob transform define the recursion for $\beta_t$ as
\bea
\beta_t(\psi,x)= \int\cP^{(0)}(d\psi',dx_t\mid \psi,x) \beta_{t+dt}(\psi',x')\qquad \beta_T(\psi,x_T)=\|\psi\|^2\label{betarecursion}
\eea
Define the effect operator $E_t$ implicitly as 
\bea
\beta_t(\psi,x)=q_t \Tr\left(E_t(x)P_\psi\right)\qquad q_t=\cP^{(0)}(x_{t:T}|x_{0:t})  \label{betaE}
\eea
The end condition $\beta_T(\psi,x_T)=\|\psi\|^2$ for all $\psi$ implies $E_T=I$.

Let \(\mathcal M_{t,dx_t}\) denote the super operator for the one-step map associated with the filtering equation Eq.~\ref{eq:rho-filter} for $\eta=0$:
\beaa
\sigma_{t+dt} =\sigma_t +d\sigma_t = \cM_{t,dx_t}(\sigma_t) \qquad \mathcal M_{t,dx_t}(A)
=A+\cL_L(A)\,dt
+\mathcal K_{i,t}(A)C^{-1}_{ij}dY_j+o(dt).
\label{eq:one-step-map-expansion}
\eeaa
where we define the super operator $\cK_{t,dx_t}(A) =    \Gamma_{ia}L_aA+\mathrm{h.c.}$ and $dY_{j,t} = dx_{j,t}-f_{j,t}^cdt$. 
$\cM_{t,dx_t}$
selects the realized increment $dx_t$, averages over unresolved quantum noise, and retains the likelihood weight of the increment.

\begin{lemm}
The backward recursion Eq.~\ref{betarecursion} and the implicit definition of $E_t$ in Eq.~\ref{betaE} imply that the effect operator $E_t$ satisfies the backward recursion
\begin{equation}
E_t(x)
=\mathcal M_{t,dx_t}^\dagger
\left(E_{t+dt}(x')\right),
\qquad
E_T=I
\label{eq:effect-recursion-section_a}
\end{equation}
The effect is related to the normalized backward message Eq.~\ref{bmessage} as
\bea
b_t(\phi,x)=\cP^{(1)} (x_{t:T}\mid \phi, x) = q_t \Tr(E_t(x)P_{\phi}).\label{bE_a}
\eea
\end{lemm}
\begin{proof}
 When $\sigma_t$ is a pure state  $\sigma_t=P_\psi$ one can write \footnote{
\beaa
\cP(\psi'|x_{0:t+dt})&=& \frac{\cP(\psi',x_{0:t},dx_t)}{\cP(x_{0:t+dt})}=\frac{p(x_{0:t})}{\cP(x_{0:t+dt})} \cP(\psi',dx_t|x_{0:t})
=\frac{\cP(x_{0:t})}{\cP(x_{0:t+dt})}\int d\psi \cP(\psi',dx_t\mid \psi,x_t) \cP(\psi\mid x_{0:t}) =\frac{\cP(x_{0:t})}{\cP(x_{0:t+dt})} \cP(\psi',dx_t\mid \psi,x_t) 
\eeaa
where we used $\cP(\psi|x_{0:t})$ is a delta function centered on $\psi$. The result follows because of the definition of $\sigma_{t+dt}$
\beaa
\sigma_{t+dt}&=&\cM_{t,dx_t}(P_\psi) =\int d\psi' P_{\psi'} \cP^{(0)}(\psi'|x_{0:t+dt})
\eeaa
}
\bea
\cM_{t,dx_t}(P_\psi) =\frac{\cP^{(0)}(x_{0:t})}{\cP^{(0)}(x_{0:t+dt})} \int  \cP^{(0)}(d\psi',dx_t\mid \psi,x_t) P_{\psi'}
\eea
Define the adjoint map $\cM^{\dagger}_{t,dx_t}$ so that for any two operators $A,B$ we have 
\bea
\Tr\left(\cM^{\dagger}_{t,dx_t}(A) B\right)=\Tr\left(A \cM_{t,dx_t} (B)\right)
\eea
The recursive relation for $\beta_t$ Eq.~\ref{betarecursion} implies a recursive relation for $E_t$.
Note that Eq.~\ref{betaE} holds for $t=T$.
Assume that Eq.~\ref{betaE} holds at $t+dt$, then
\beaa
\beta_t(\psi,x)
&=& q_{t+dt}\int \cP^{(0)}(d\psi',dx_t \mid \psi,x) \Tr\left(E_{t+dt}(x')P_{\psi'}\right)=q_t \Tr\left(E_{t+dt}(x')\cM_{t,dx_t}(P_{\psi})\right)\\
&=&q_t\Tr\left(\cM_{t,dx_t}^{\dagger}(E_{t+dt}(x'))P_{\psi}\right)=q_t \Tr\left(E_t(x) P_{\psi}\right)
\label{eq:beta-cancel-likelihood}
\eeaa
where in the last step we defined the retrograde recursion. 
%Since $\cM_{t,dx_t}$ is given by Eq.~\ref{eq:one-step-map-expansion}, its adjoint is given by the adjoints of Eqs.~\eqref{lindblad} and
%\eqref{eq:K-forward-section} as
%\begin{align}
%\cL_L^\dagger(E)
%&=i[H_t,E]
%+Q_{ab}\left(
%L_{b,t}^\dagger E L_{a,t}
%-\frac{1}{2}\{L_{b,t}^\dagger L_{a,t},E\}
%\right),
%\label{eq:L-adjoint-section}
%\\
%\mathcal K_{i,t}^\dagger(E)
%&=\Gamma_{ia}^{*}L_{a,t}^\dagger E
%+\Gamma_{ia}E L_{a,t}.
%\label{eq:K-adjoint-section}
%\end{align}
%Using the expansion Eq.~\eqref{eq:one-step-map-expansion}, the exact recursion
%Eq.~\eqref{eq:effect-recursion-section} gives
%\begin{align}
%E_t
%={}&E_{t+dt}
%+\mathcal L_t^\dagger(E_{t+dt})\,dt
%+\mathcal K_{i,t}^\dagger(E_{t+dt})
%C^{-1}_{ij}dY_j.
%\label{eq:effect-backward-discrete}
%\end{align}
Therefore Eq.~\ref{betaE} holds for all $t$, which completes the proof of the first statement. 
From Eq.~\ref{betaE} and the relation $\beta_t(\psi,x) = \|\psi\|^2b_t(\phi,x)$ with $\phi=\psi/\|\psi\|$ we obtain
\beaa
\cP^{(1)}(x_{t:T}\mid \phi,x)=\frac{\beta_t(\psi,x)}{\|\psi\|^2} = \frac{q_t}{\|\psi\|^2} \Tr(E_t(x)P_{\psi}) = q_t \Tr(E_t(x)P_{\phi}) 
\eeaa
which proves Eq.~\ref{bE_a}. 
 \end{proof}
 
For completeness, the retrograde filter equation Eq.~\ref{eq:effect-recursion-section_a} can be written explicitly as 
\begin{equation}
E_t=\mathcal M_{t,dY}^\dagger(E_{t+dt})
=
E_{t+dt}
+\mathcal L^\dagger(E_{t+dt})\,dt
+\mathcal K_i^\dagger(E_{t+dt})\,C^{-1}_{ij}\,dY_j
+o(dt).
\end{equation}
The adjoint Lindblad generator is
\begin{equation}
\mathcal L^\dagger(E)
=
i[H,E]
+
Q_{ab}
\left(
L_b^\dagger E L_a
-\frac12\{L_b^\dagger L_a,E\}
\right).
\end{equation}
The adjoint observation superoperator is $\cK^\dagger(E) = \Gamma_{ia} EL_a +h.c.$. 
Therefore, explicitly,
\begin{equation}
\boxed{
\begin{aligned}
E_t
={}&E_{t+dt}
+i[H,E_{t+dt}]\,dt\\
&+
Q_{ab}\left[
L_b^\dagger E_{t+dt}L_a
-\frac12
\left(
L_b^\dagger L_aE_{t+dt}
+
E_{t+dt}L_b^\dagger L_a
\right)
\right]dt\\
&+
\left(
\Gamma_{ia}E_{t+dt}L_a
+
\Gamma_{ia}^*L_a^\dagger E_{t+dt}
\right)
C^{-1}_{ij}\,dY_j
+o(dt).
\end{aligned}
}\label{explicit}
\end{equation}

Since $E_t$ and $\sigma_t$ evolve as adjoints, we have $\Tr(E_t \sigma_t) =\Tr(\cM^\dagger(E_{t+dt}) \sigma_t))=
\Tr(E_{t+dt} \sigma_{t+dt})$. 
Therefore 
\beaa
\Tr(E_t \sigma_t)=\Tr(E_T \sigma_T) = \Tr(\sigma_T) =\Lambda_T
\eeaa 
where the last step follows from Eq.~\ref{eq:TRsigma}.

From Eq.~\ref{doob_transform} the conditional physical and
linear path measures are related by 
\[
 \cP^{(1)}(d\phi \mid x_{0:t})
 =
 \frac{\|\psi\|^2}{\Lambda_t}
 \cP^{(0)}(d\psi \mid x_{0:t}).
\]
with $\Lambda_t=\int \|\psi\|^2 \cP^{(0)}(d\psi\mid x_{0:t}) = \Tr (\sigma_t)$. 
Since $\beta_t(\psi,x)=\|\psi\|^2b_t(\hat \psi,x)$, we obtain
\[
\begin{aligned}
p_t:=\cP^{(1)}(x_{t:T}\mid x_{0:t})=\int b_t(\phi,x_t)\,
   \cP^{(1)}(d\phi \mid x_{0:t})=\frac{1}{\Lambda_t}
  \int \beta_t(\psi,x_t)\,
  \cP^{(0)}(d\psi\mid x_{0:t})=\frac{q_t}{\Lambda_t} \Tr(E_t \sigma_t)
  \end{aligned}
\]
Hence $\Tr(E_t \sigma_t) = \Lambda_t \ell_t$ where we define $\ell_t=\frac{p_t}{q_t}$. 
Therefore we obtain
\beaa
\Tr(E_t \sigma_t) =\Lambda_t \ell_t = \Lambda_T\qquad \Tr(E_t \rho_t) = \frac{\Tr(E_t\sigma_t)}{\Lambda_t} = \ell_t
\eeaa
As a sanity check, we compute 
\beaa
p_t=\cP^{(1)}(x_{t:T}|x_{0:t})  =\int d\phi_t \cP^{(1)}(\phi_t\mid x_{0:t})\cP^{(1)}(x_{t:T}\mid \phi_t,x_t)=q_t \Tr(\rho_tE_t(x_t))=q_t\ell_t
\eeaa

\end{document}